\documentclass[11pt,a4paper]{article}

\usepackage[british]{babel}

\usepackage[a4paper,top=2cm,bottom=2cm,left=2.5cm,right=2.5cm,marginparwidth=1.75cm]{geometry}

\usepackage{amsmath,amsthm,amsfonts,amssymb,amscd, fancyhdr, color, comment, graphicx, environ, mathtools, cancel, float, mathrsfs}

\usepackage[colorlinks=true, allcolors=blue]{hyperref}
\usepackage{hyperref}
\usepackage[font=small,labelfont=bf]{caption} 

\usepackage{enumitem}
\usepackage{chngcntr}
\usepackage{authblk}
\usepackage[T1]{fontenc}    % Font encoding
\usepackage{csquotes}       % Include csquotes

\newtheorem{definition}{Definition}[section]
\newtheorem{theorem}{Theorem}[section]
\newtheorem{conjecture}{Conjecture}[section]

\newtheorem{lemma}[theorem]{Lemma}

\newtheorem{claim}{Claim}[theorem]
\newtheorem{subclaim}{Claim}[claim]
\newtheorem{claimout}[theorem]{Claim}
\newtheorem{fact}[definition]{Fact}

\newcommand{\E}{\mathbb E}
\newcommand{\I}{\mathbb I}
\newcommand{\N}{\mathbb N}

\newcommand{\deff}{\coloneqq}
\DeclareMathOperator{\negl}{negl}
\DeclareMathOperator{\poly}{poly}

\newcommand{\unique}{\mathrm{unique}}
\newcommand{\Dist}{\mathrm{Dist}}
\newcommand{\Good}{\mathrm{Good}}
\newcommand{\ba}{\mathbf{a}}
\newcommand{\bx}{\mathbf{x}}
\newcommand{\by}{\mathbf{y}}
\newcommand{\bxp}{\mathbf{x'}}
\newcommand{\bxpp}{\mathbf{x''}}
\newcommand{\byp}{\mathbf{y'}}

\newcommand{\byl}{\mathbf{y}_{<}}
\newcommand{\byr}{\mathbf{y}_{>}}

\newcommand{\ket}[1]{|#1\rangle}
\newcommand{\bra}[1]{\langle#1|}
\newcommand{\ketbra}[2]{\ket{#1}\!\bra{#2}}
\newcommand{\braket}[2]{\bra{#1}\!\ket{#2}}
\newcommand{\norm}[1]{\|#1\|}
\newcommand{\abs}[1]{|#1|}
\newcommand{\Tr}{\mbox{\rm Tr}}
\newcommand{\surj}{\mathrm{Surj}}

\newcommand{\romi}[1]{\textcolor{red}{[Romi: #1]}}
\newcommand{\thomas}[1]{\textcolor{blue}{[Thomas: #1]}}

\renewcommand{\romi}[1]{}
\renewcommand{\thomas}[1]{}

\title{PRS Length Expansion}
\date{}

\author[1,*]{Romi Levy}
\author[2,1,**]{Thomas Vidick}

\affil[1]{\small Computer Science Department, Weizmann Institute of Science}
\affil[2]{\small School of Computer and Communication Sciences, Ecole Polytechnique Fédérale de Lausanne}

\affil[*]{Email: \texttt{romi.levy@weizmann.ac.il}}
\affil[**]{Email: \texttt{thomas.vidick@epfl.ch}}

\begin{document}
\maketitle

\begin{abstract}
One of the most fundamental results in classical cryptography is that the existence of Pseudo-Random Generators (PRG) that expands $k$ bits of randomness to $(k+1)$ bits that are pseudo-random implies the existence of PRG that expand $k$ bits of randomness to $k+f(k)$ bits for any $f(k)=poly(k)$.

It appears that cryptography in the quantum realm sometimes works differently than in the classical case. Pseudo-random quantum states (PRS) are a key primitive in quantum cryptography, that demonstrates this point. There are several open questions in quantum cryptography about PRS, one of them is --- can we expand quantum pseudo-randomness in a black-box way with the same key length? Although this is known to be possible in the classical case, the answer in the quantum realm is more complex. This work conjectures that some PRS generators can be expanded, and provides a proof for such expansion for some specific examples.  In addition, this work demonstrates the relationship between the key length required to expand the PRS, the efficiency of the circuit to create it and the length of the resulting expansion. 
\end{abstract}


%-------------------------------------------
% Paper Body
%-------------------------------------------

%--- Section ---%
\newpage
\begin{center}

\tableofcontents

\end{center}

%--- Section ---%
\newpage
\section{Introduction}

Pseudo-random quantum states (PRS) were first introduced by Ji et al.~\cite{ji2018pseudorandom} in 2018. They caught much attention thanks to Kretschmer's result~\cite{kretschmer2023quantum} suggesting they might be weaker than one way functions (OWF).
This was a complete surprise since  PRS are a natural quantum notion of pseudo-randomness. In contrast the classical notion of pseudo-randomness, Pseudo-Random Generators (PRG), are known to be an equivalent assumption to OWF.

This raises the question of which properties of pseudo-randomness are maintained when moving to the quantum realm.
There are several open questions in quantum cryptography about PRS, one of them is --- can we expand quantum pseudo-randomness in a black-box way with the same key length? Although this is known to be possible in the classical case, the answer in the quantum realm is more complex. 

In~\cite{muguruza2024quantum} Muguruza et al.\ demonstrate that PRS cannot be shrunk in a black-box way (unlike PRG) in some regime of parameters. Here shrinking means the reduction of the number of output qubits. More specifically, they show that relative to some oracle, neither pseudo-deterministic OWF (PD-OWF) nor PD-PRG exist, while ``long'' PRS persist.
By ``long'' PRS we mean a PRS that has $\poly(\lambda)$ or more qubits, where $\lambda$ is the security parameter.
These notions (PD-OWF and PD-PRG), introduced by Ananth et al.~\cite{ananth2023pseudorandom}, are shown to be implied by ``short'' PRS (a PRS with $\Theta(\log\lambda)$ many qubits for sufficiently large implied constant). Thus, \cite{muguruza2024quantum} reveals a scenario where ``long'' PRS exist in the presence of an oracle, yet ``short'' PRS do not, implying the absence of a black-box method for shrinking PRS in this regime.
This result, combined with the finding by Brakerski and Shmueli~\cite{brakerski2020scalable} that ``very short'' PRS (e.g., at most $O(\log\lambda)$ for a sufficiently small implied constant) exist unconditionally, implies that expanding \textit{any} PRS in a fully flexible, black-box manner is not possible.
For example, it does not seem possible to go from a constant, or even sub-logarithmic, length PRS to one that is super-logarithmic. Yet one may still hope to expand once the number of qubits is sufficiently large, e.g. $n=\poly(\lambda)$ or even $n=\omega(\log\lambda)$. Note that if it is indeed possible, such an expansion would have to use a key as an additional source of randomness. This is because if not, then the expansion process can be efficiently reversed, leading to an efficient algorithm to distinguish the resulting state from a Haar random state.
Additionally, Brakerski and Shalit~\cite{BS24} observed that any unitary algorithm expanding a PRS must rely on some cryptographic assumption\thomas{Can we spell this out a tiny bit? I forgot what they do...Is it that one needs to make an assumption, beyond security of the original PRS? Would it be accurate to reformulate as ``observed that any unitary algorithm, or keyed family of algorithms, that expands any PRS must rely on an additional cryptographic assumption''? I'm confused how that relates to our proof of expansion, which uses no assumption other than the PRS (but is for a specific PRS--is that it?)}.

In this work we conjecture that some\footnote{Any PRS that satisfy the condition stated in \ref{condition}.} PRS can be expanded in a black-box manner, and provide a proof for such expansion for a specific, widely referred to, family of PRS introduced in \cite{ji2018pseudorandom, brakerski2019pseudo}.  This result is formally stated as Theorem~\ref{thm1} in Section~\ref{sec2}. 
Our proof uses the purification technique introduced by  Ma in a~\href{https://www.youtube.com/playlist?list=PLgKuh-lKre10dVLxTTK-wcEOD-FTe_SVq}{workshop on pseudorandom unitaries} given by him and Huang at the Simons Institute, and later published in~\cite{ma2024construct}.
Our proof is delicate and requires several approximations in order to get to the desired point; the main technical difficulty is posed by the re-use of the key, which precludes more standard hybrid arguments such as are used in e.g.~\cite{schuster2024random}. 
In addition, our work explores the relationship between the key length required to expand the PRS, the efficiency of the algorithm that expands the PRS, and the length of the resulting expansion.

Recent work has focused on understanding PRS and its role in quantum cryptography, including various construction methods \cite{ji2018pseudorandom, brakerski2019pseudo, brakerski2020scalable}, implications \cite{ananth2022cryptography}, and potential applications such as quantum money \cite{ji2018pseudorandom, ananth2022cryptography}. PRS may represent a weaker assumption than OWF, allowing some cryptographic applications based on PRS to remain valid even if OWF does not exist.
Therefore understanding what manipulations can be applied to the length of the output while retaining the PRS properties is of high importance, and can be useful for further applications.

\paragraph{Acknowledgments.} \thomas{moved this up and added funding info}
R.L.\ extends deep gratitude to Nir Magrafta, Itay Shalit, and Assaf Harel for their valuable help in the preparation of this work. T.V.\ is supported by a research grant from the Center for New Scientists
at the Weizmann Institute of Science and AFOSR Grant No. FA9550-22-1-0391.

%--- Section ---%
\section{Preliminaries}

\subsection{Notation}

We use bold font $\ba,\bx$ to denote tuples $\ba=(a_1,\ldots,a_t)$, $\bx=(x_1,\ldots,x_t)$. We write $\ba\circ\bx$ for concatenation of tuples, i.e. $\ba\circ\bx=(a_1,\ldots,a_t,x_1,\ldots,x_t)$ as a $2t$-tuple. We let  $\mathfrak{S}_{t}$ be the set of permutations on $t$ elements (usually denoted $\{1,\ldots,t\}$). 

We use $\lambda$ to refer to the security parameter in cryptographic constructions. For $n$ an integer, we often use the notation $N=2^n$ and $\omega_N = e^{\frac{2\pi i}{N}}$.

For a Hilbert space $\mathcal{H}$, we let $S(\mathcal{H})$ denote the set of pure quantum states on $\mathcal{H}$, $D(\mathcal{H})$  the set of density operators
and $U(\mathcal{H})$ the set of unitary operators.

We use the abbreviation i.i.d.\ to mean independent and identically distributed.

\subsection{Quantum information}

\begin{definition}[Trace Distance]
    The trace distance of two quantum states $\rho_0, \rho_1 \in D(\mathcal{H})$ is
    \begin{equation*}
        TD\left( \rho_0, \rho_1\right) \deff \frac{1}{2}\norm{\rho_0 - \rho_1}_1\;.
    \end{equation*}
\end{definition}
\begin{fact}
\label{fact1}
    When $\rho_0 = \ketbra{\psi_0}{\psi_0}, \rho_1 = \ketbra{\psi_1}{\psi_1}$ are pure states we get 
    $TD\left( \rho_0, \rho_1\right) = \sqrt{1 - \abs{\braket{\psi_0}{\psi_1}}^2}$.
\end{fact}

\begin{fact}
\label{fact2}
    For every trace-preserving-completely-positive map $\Phi$,
    \begin{equation*}
        TD\left( \Phi(\rho_0), \Phi(\rho_1) \right) \le TD \left( \rho_0, \rho_1 \right)\;.
    \end{equation*}
\end{fact}

\begin{definition}[Partial Trace of a Pure State]
\label{def1}
    Let $\ket{\psi}_{A,E}$ be some pure state on systems $A,E$\\
    For simplicity we use the following notation:
\begin{equation*}
    \Tr_E\left[\ket{\psi}_{A,E} \right] \deff 
    \Tr_E\left[\ketbra{\psi}{\psi}_{A,E}\right]
\end{equation*}
\end{definition}

\begin{claimout}
\label{clm1}
    Let $\rho$ be some density matrix and $\ket{\chi} = \sum_i \ket{\phi_i}\otimes \ket{\psi_i}_E$ be a purifying state, i.e.\ $\Tr_E[\ketbra{\chi}{\chi}] = \rho$. Let $A$ be any isometry.
    Then $\ket{\chi'} = \sum_i \ket{\phi_i}\otimes  A\ket{\psi_i}_E$ is also a purification state: $\Tr_E[\ketbra{\chi'}{\chi'}] = \rho$.
\end{claimout}

\begin{proof}
Expanding using the definition of $\ket{\chi}$ and using linearity of the partial trace,
\begin{align*}
    \rho &= \Tr_E[\ketbra{\chi}{\chi}]\\
    & = \Tr_E\Big[\sum_i \ket{\phi_i}\otimes \ket{\psi_i}_E \cdot \sum_j \bra{\phi_j}\otimes \bra{\psi_j}_E\Big]\\
    & = \Tr_E\Big[\sum_{i,j} \ketbra{\phi_i}{\phi_j}\otimes \ketbra{\psi_i}{\psi_j}_E\Big]\\
    & = \sum_{i,j} \ketbra{\phi_i}{\phi_j}\otimes \Tr[\ketbra{\psi_i}{\psi_j}]\\
    & = \sum_{i,j} \ketbra{\phi_i}{\phi_j} \cdot \langle{\psi_j}\ket{\psi_i}\;.
\end{align*}
On the other hand, 
\begin{align*}
    \Tr_E[\ketbra{\chi'}{\chi'}] &= \Tr_E\Big[\sum_i \ket{\phi_i}\otimes  A\ket{\psi_i}_E \cdot \sum_j \bra{\phi_j}\otimes  \bra{\psi_j}_EA^\dagger\Big]\\
    & = \Tr_E\Big[\sum_{i,j} \ketbra{\phi_i}{\phi_j} \otimes  A\ket{\psi_i}_E  \bra{\psi_j}_EA^\dagger\Big]\\
    & = \sum_{i,j} \ketbra{\phi_i}{\phi_j} \otimes  \Tr[A\ket{\psi_i}_E  \bra{\psi_j}_EA^\dagger]\\
    & = \sum_{i,j} \ketbra{\phi_i}{\phi_j} \cdot  \bra{\psi_j} A^\dagger \cdot A\ket{\psi_i}\;.
\end{align*}
Because $A$ is an isometry, $\bra{\psi_j} A^\dagger \cdot A\ket{\psi_i} = \langle{\psi_j} \ket{\psi_i}$.
This finishes the proof. 
\end{proof}

\begin{claimout}
    \label{clm3}
    Let $\mathcal{H}$ be a Hilbert space, $\mathcal{B}$ an orthonormal basis for $\mathcal{H}$ and $U \in U(\mathcal{H})$. Then
    \begin{equation*}
        \sum_{x\in \mathcal{B}}  U\ket{x}\otimes\ket{x} = \sum_{x\in \mathcal{B}}  \ket{x}\otimes U^T\ket{x}\;.
    \end{equation*}
\end{claimout}

\begin{lemma}[Gentle measurement lemma~\cite{ogawa2002new}]\label{lem:gentle}
For all positive semidefinite $\rho$ and $X$, 
\[ TD\big( \rho , \sqrt{X}\rho\sqrt{X}\big)\,\leq\, 2\sqrt{\Tr(\rho)}\sqrt{\Tr(\rho(\I-X))}\;.\]
\end{lemma}

\subsection{Cryptography}

\begin{definition}[Negligible Function]
    A function $\epsilon(\lambda)$ is negligible if
    for all constant $c > 0$ there exists some $\Lambda$ s.t: $\forall \lambda > \Lambda:$
    \begin{equation*}
        \epsilon(\lambda) < \lambda^{-c}\;.
    \end{equation*}
\end{definition}

\begin{definition}[Quantum-Secure PRF]\label{def:QPRF}

Let $\mathcal{K, X, Y}$ be the key space, the domain and range, all implicitly depending on the security parameter $\lambda$. 
\newline A keyed family of functions $\{ PRF_k : \mathcal{X \rightarrow Y} \}_{k\in \mathcal{K}}$
is a quantum-secure pseudo-random function (Quantum-Secure PRF) if for any
polynomial-time quantum oracle algorithm A, $PRF_k$ with a random $k \leftarrow \mathcal{K}$ is
indistinguishable from a truly random function $f \leftarrow \mathcal{Y^X}$
in the sense that:

\begin{equation*}
    \abs{\underset{k \xleftarrow{\$}\mathcal{K} }{Pr}[\mathcal{A}^{PRF_k}(1^\lambda) = 1] - 
    \underset{f \xleftarrow{\$} \mathcal{Y^X} }{Pr}[\mathcal{A}^{f}(1^\lambda) = 1]} = \negl(\lambda)
\end{equation*}
    
\end{definition}

\begin{definition}[PRS]\label{def:prs}
    Let $\lambda$ be the security parameter. Let $\mathcal{H}$ be a Hilbert space and $\mathcal{K}$ the key space, both parameterized by $\lambda$. A keyed family of quantum states $\{\ket{\phi_k}\in S(\mathcal{H})\}_{k\in \mathcal{K}}$ is pseudo-random, if the following two conditions hold:

    \begin{enumerate}
        \item \textbf{(Efficient generation).} There is a polynomial-time quantum algorithm $G$ that generates state $\ket{\phi_k}$ on input $k$. That is, for all $k\in \mathcal{K}: G(k) = \ket{\phi_k}$.

        \item \textbf{(Pseudo-randomness).} Any polynomially many copies of $\ket{\phi_k}$ with the same random $k\in \mathcal{K}$ is computationally indistinguishable from the same number of copies of a Haar random state. 
    More precisely, for any efficient quantum algorithm $\mathcal{A}$ and any $m \in poly(\lambda)$, 
    \[\abs{\underset{k \xleftarrow{\$}\mathcal{K} }{Pr}[\mathcal{A}(\ket{\phi_k}^{\otimes m}) = 1] - 
    \underset{\ket{\psi} \xleftarrow{} \mu }{Pr}[\mathcal{A}(\ket{\psi}^{\otimes m}) = 1]} = \negl(\lambda)\;.\]

    where $\mu$ is the Haar measure on $S(\mathcal{H})$.
    \end{enumerate}

\end{definition}

Since this definition was first given, in~\cite{ji2018pseudorandom}, several other constructions were proposed.
The first one was given in the same paper \cite{ji2018pseudorandom}, and we will refer it as the general-phase PRS.

\begin{definition}[General Phase PRS on $n$ qubits] \label{GP PRS def} For any $n=n(\lambda)$ such that $n=\omega(\log\lambda)$ and $n=\poly(\lambda)$, 
\[ \ket{\phi_k}\,=\, 2^{-\frac{n}{2}} \sum_{x\in \{0,1\}^n} \omega_N^{f_k(x)} \ket{x}\;,\]
    Where $\{f_k\}_k$ is a quantum secure family of pseudo-random functions (PRF), $N = 2^n$ and $\omega_N = 2^{2i\pi/N}$.
\end{definition}

Another suggested construction, which we will refer as the binary-phase PRS, was given in~\cite{brakerski2020scalable}.

\begin{definition}[Binary Phase PRS on $n$ qubits] \label{BP PRS def}
For any $n=n(\lambda)$ such that $n=\omega(\log\lambda)$ and $n=\poly(\lambda)$, 
    \[\ket{\phi_k}\,=\, 2^{-\frac{n}{2}} \sum_{x\in \{0,1\}^n} (-1)^{f_k(x)} \ket{x}\;,\]
    where $\{f_k\}_k$ is a quantum secure PRF.
\end{definition}

Because our results depend on the specific algorithm used to generate the PRS, we here recall the natural implementation of both phase PRS constructions. \thomas{added this since our results depend on it}

\begin{definition}[Algorithm for generating phase PRS]\label{def:prs-algo}
Let $\ket{\phi_k}$ be the binary phase (resp.\ general phase) PRS construction defined above. Then $\ket{\phi_k}=PRS_k\ket{0^n}$, where $PRS_k = U_{f_k} \circ QFT_n$, with $QFT_n= H^{\otimes n}$, $H$ the single-qubit Hadamard (resp.\ $QFT_n$ the quantum Fourier transform over $\mathbb{F}_N$, $N=2^n$, i.e. $QFT_n \ket{x} = \sum_y \omega_N^{xy}\ket{y}$) and $U_{f_k}\ket{x}=(-1)^{f_k(x)}\ket{x}$ (resp.\ $U_{f_k}\ket{x} = \omega_N^{f_k(x)} \ket{x}$).
\end{definition}

We define what we mean by length expansion of a PRS. 

\begin{definition}[valid length expansion]
    Let $\{C_k\}_k$ be some PRS on $n$ qubits.
    We say that a construction is a valid length expansion for $\{C_k\}_k$  if when instantiating the PRS calls within the construction with $\{C_k\}_k$, the resulting algorithm produces a PRS on $\ell(n) > n$ qubits.
\end{definition}

%--- Section ---%
\section{Our results}\label{sec2}

We present three constructions for expanding PRS, referred to as Constructions 1, 2, and 3. We prove that Construction 1 is valid for the binary-phase PRS and propose that it is similarly applicable to the general-phase PRS. Additionally, we introduce two variations of this construction (Constructions 2 and 3) and conjecture their validity for both binary- and general-phase PRS. Finally, we propose a Generalization Condition, which we conjecture allows any PRS satisfying this condition to be expanded using any of our constructions.

\subsection{The Expansion Construction}\label{construction1}

Here we introduce a construction that expands the binary-phase PRS and the general-phase PRS and can be potentially used to expand some\footnote{Any PRS that satisfy the condition stated in \ref{condition}.} other PRS as a black box to generate longer PRS. We refer to this construction as ``Construction 1.'' The construction is depicted in Figure~\ref{fig:construction1}. It is very similar to the standard approach for expanding the output length of a PRG, and proceeds by composing input-shifted instances of the PRS generation algorithm, which we view as a unitary map $PRS_k$ parameterized by the key $k$.

\begin{center}
    \includegraphics[width=0.6\linewidth]{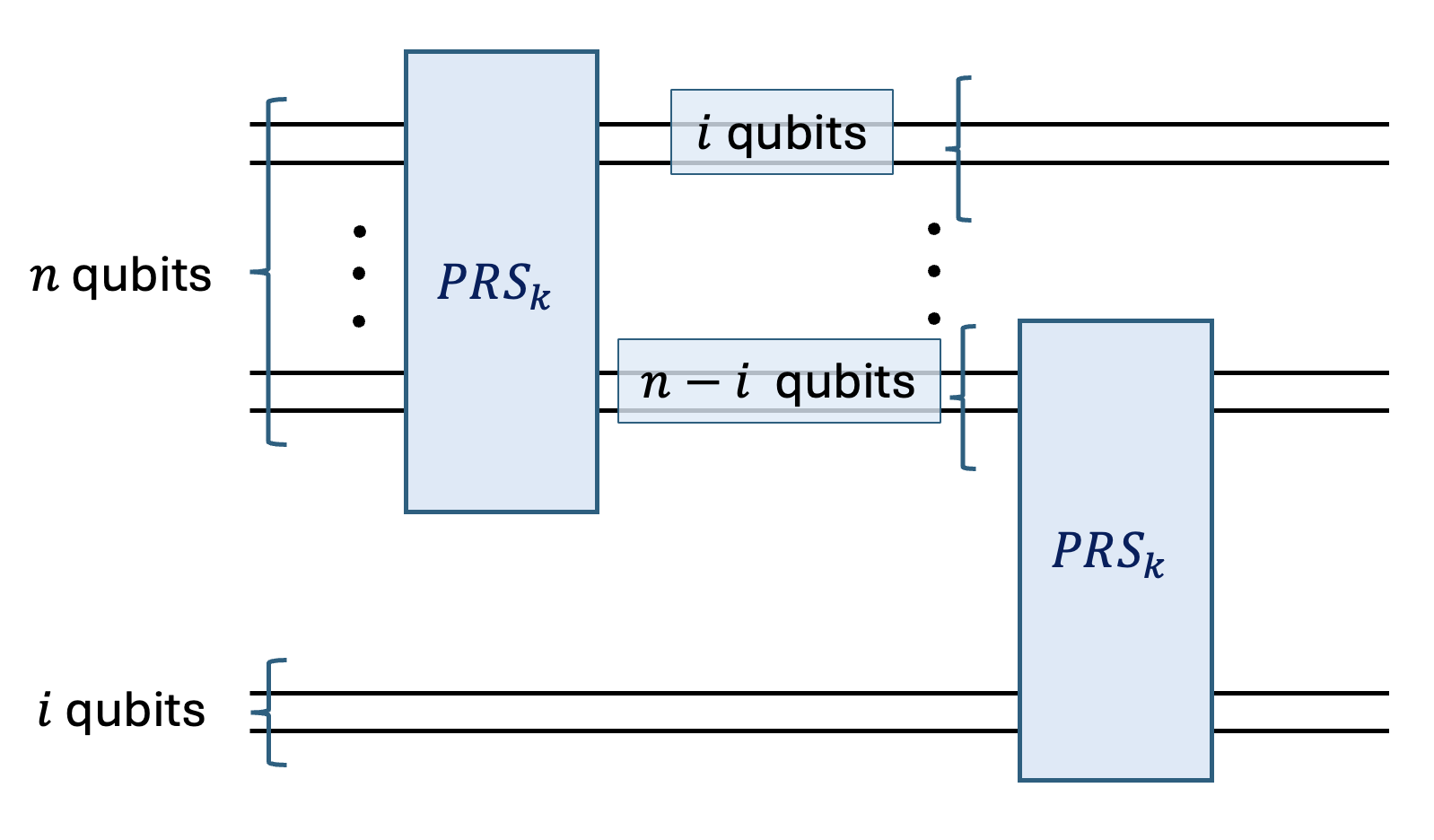}
    \captionof{figure}{\label{fig:construction1}PRS expansion with 2 blocks.}
\end{center}

We will later claim and prove formally that the output of this construction is a PRS on $n + i$ qubits, as long as $n-i=\omega(\log \lambda)$ and $PRS_k$ is a specific construction. Note that this will be true {if and only if} we can apply any efficient unitary on the output and still get a PRS on $(n + i)$ qubits. In particular when this unitary is $H$---More formally:
\begin{claimout}
    Let $\lambda$ be an asymptotic security parameter. Let $\{C_k\}_k$ be a family of quantum circuits indexed by $k=k(\lambda)$ on $n=n(\lambda)$ qubits and $U=(U_\lambda)$ an efficiently generated family of unitaries on $n(\lambda)$ qubits.
    Then $\{C_k\ket{0}^{\otimes n}\}_k$ is a PRS $\iff$ $\{UC_k\ket{0}^{\otimes n}\}_k$ is a PRS. 
\end{claimout}

The claim follows since otherwise applying $U$ can be used to distinguish the PRS before $U$ from a Haar random state, that remains Haar random after applying $U$, in contradiction to the definition of PRS.
So from now on for convenience we shall analyze the following construction:
\begin{center}\label{fig:construction11}
    \includegraphics[width=0.6\linewidth]{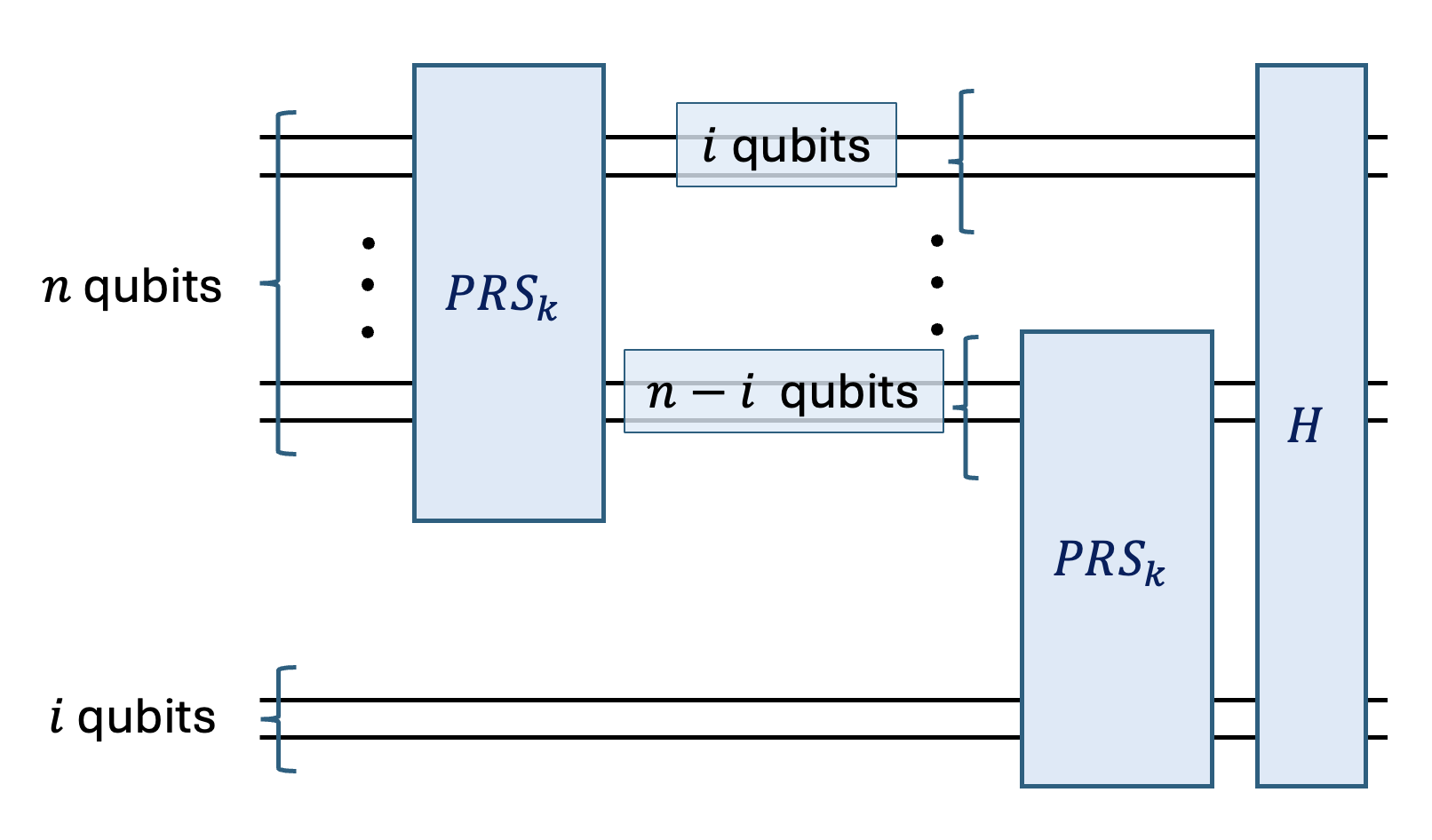}
    \captionof{figure}{Modified PRS expansion with 2 blocks.}
\end{center}

Let $\ket{\psi_k}$ denote the state at the end of this construction. The following is our main result. 

\begin{theorem}
\label{thm1}
    Let $\lambda$ be the security parameter, $PRS_k$ be the binary-phase PRS as in definition \ref{BP PRS def}, implemented as described in Definition~\ref{def:prs-algo}, $n(\lambda)$ the number of qubits of the output state of $PRS_k$ such that $n = \poly(\lambda)$ and let $i=i(n) : \N \rightarrow \N$ be such that $n - i(n) = \omega(\log(\lambda))$.
   Then the family of states $\{\ket{\psi_k}\}_k$ from ~\hyperref[fig:construction11]{Figure 2} is a PRS on $n + i$ qubits.
\end{theorem}

Theorem~\ref{thm1} is proved in Section~\ref{sec:proof1}. We note that two elements distinguish the result from known results. Firstly, it was recently shown in~\cite{schuster2024random} that ``gluing'' two pseudorandom \emph{unitaries} (PRU) along a similar scheme as that described in Figure~\ref{fig:construction1} does produce a valid, larger pseudorandom unitary. This result is incomparable because PRUs are a different assumption than PRS (PRUs trivially imply PRS, but the reverse is not known). In particular, the proof uses properties of random unitary operations; in our construction, either constituent or the whole may all be very far from random as unitary operations. Secondly, in our situation we are able to analyze the construction when both instantiations of the PRS are based on the same key. This precludes hybrid arguments of the kind that is essential to~\cite{schuster2024random} and requires a more hands-on analysis. While the result is more specific---it does not apply to all PRS---conceptually it provides the first step towards establishing more general length manipulation techniques on PRS that would parallel the classical setting.

\paragraph{Notes:}
\begin{itemize}
    \item A similar theorem, where the binary-phase PRS is replaced with the general-phase PRS (given in Definition \ref{GP PRS def}), can be proven using a similar proof, where $H$ is replaced with the quantum Fourier transform where needed. We omit the details. 

    \item Since the first and second blocks of our construction use the same key, it means we were able to expand this PRS without increasing the key length.
   
    \item A distinguishing feature of our result is that it allows for the expansion of PRS without compromising the number of copies that can be provided to the adversary. Specifically, our construction remains secure for any number of copies $t$ as long as $t \in \poly(\lambda)$, in accordance with Definition \ref{def:prs}. As far as we know, this is a unique aspect of our result, as previous works addressing the expansion problem either imposed a constant limit on the number of copies to begin with \cite{schuster2024random, chen2024power}, or involved a trade-off on the number of copies \cite[Theorem C.2]{gunn2023commitments} .

    \item In some cases one may care more about the efficiency of the circuit or the length of the output rather than the amount of randomness (i.e the key length). In such cases our result can be extended in such a way that the first PRS application can be made according to \emph{any} PRS, and not only the binary- (or general-)phase constructions. This follows from a simply hybrid argument, as replacing the first block with any PRS on $n-$qubits can only result in a computationally indistinguishable final state. Indeed, assume towards contradiction this is not the case. Then the construction can be used to distinguish between the PRS and $\frac{1}{\sqrt{N}}\sum_{x\in \{0,1\}^n} (-1)^{r(x)}\ket{x} $ (where $r(x)$ is a random function as in \cite{brakerski2019pseudo} Theorem 1.) which is a contradiction to the fact they are both indistinguishable from Haar and therefore indistinguishable from each other. 
    
    \item Unfortunately, we do not know how to carry out the same argument for the second block; and indeed it is unlikely that it could be done in such generality. This is simply because the definition of a PRS only guarantees how the family of circuits behaves on the specific fixed input state $\ket{0}$, not on all input states. As a result, the output on other input states may not be indistinguishable from Haar; indeed it is not hard to construct such examples artificially. For example, on input $\ket{0}$ the PRS algorithm creates the right PRS state, but on any other input it

    does nothing. One can easily verify it satisfy the definition of PRS but fails to give something interesting on inputs that are not $\ket{0}$.
\end{itemize}

\subsection{Additional Expansion Constructions}

In this section we provide additional possible length-expanding constructions for the binary/general-phase PRS.

\subsubsection{Construction 2}
\label{construction2}

\begin{center}
    \includegraphics[width=0.6\linewidth]{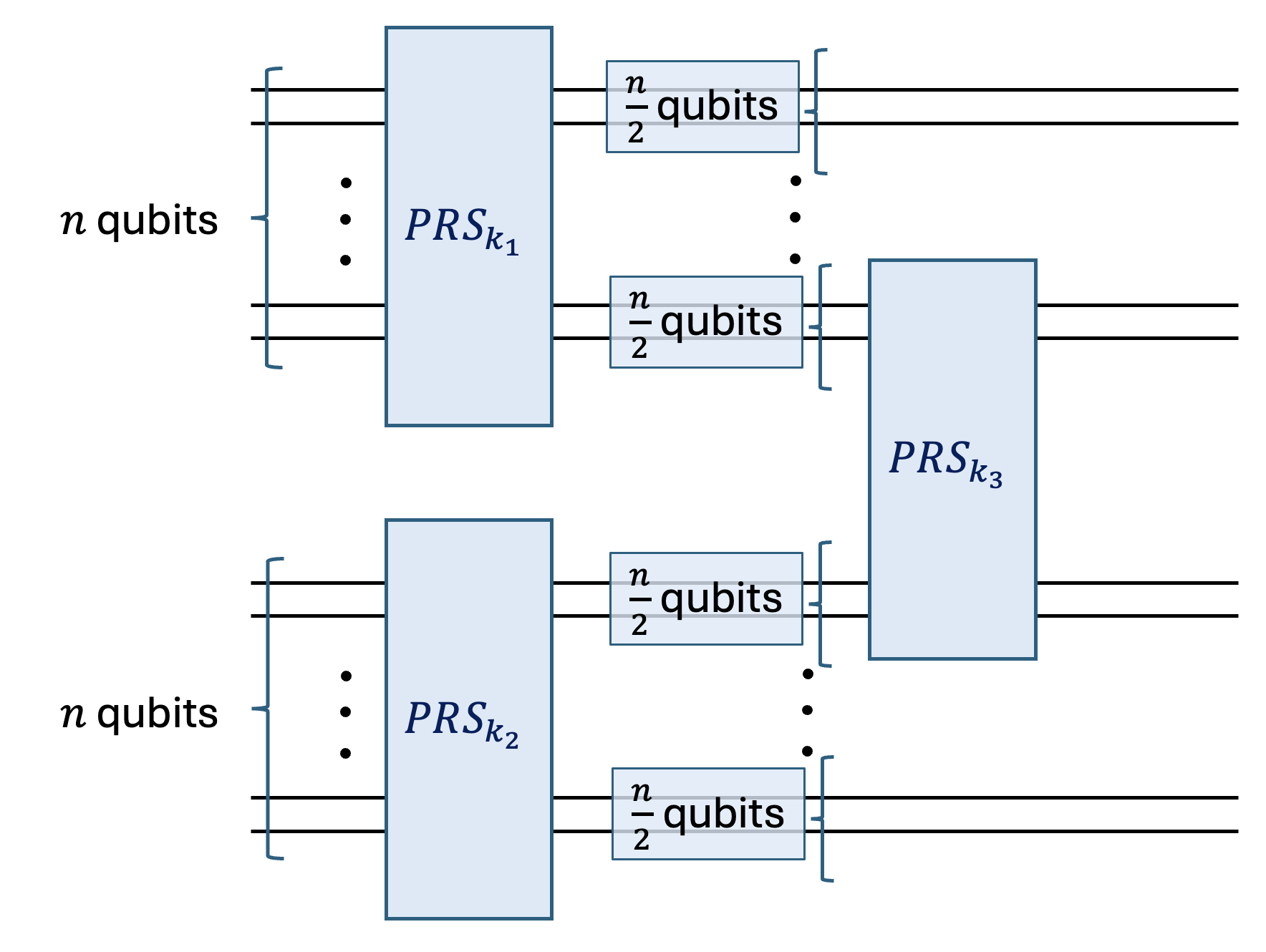}
    \captionof{figure}{\label{fig:construction2}PRS expansion with 3 blocks}
\end{center}

\begin{itemize}
    \item This construction can be proven as a valid length expansion for the binary-phase PRS, using a purification technique much like in the proof in Section \ref{sec:proof-thm1}, but we omit the details.. The purification technique was introduced by Ma in a \href{https://www.youtube.com/playlist?list=PLgKuh-lKre10dVLxTTK-wcEOD-FTe_SVq}{Workshop on pseudorandom unitaries} given by him and Huang at Simons Institute. And indeed a similar construction is presented in \cite{schuster2024random} for pseudo-random unitaries which are a potential stronger assumption than PRS.

    \item Note that here the overlap between the first and second layers is precisely $\frac{n}{2}$. 

    \item Another notable aspect of this construction is that we are able to expand the state to $2n$ qubits unlike
    \hyperref[fig:construction1]{Construction 1}, but at the expense of a key of 3 times the original length. (We do not know if this expansion construction works by re-using the same key thrice.)
\end{itemize}

\subsubsection{Construction 3}
\label{construction3}

\begin{center}
    \includegraphics[width=0.8\linewidth]{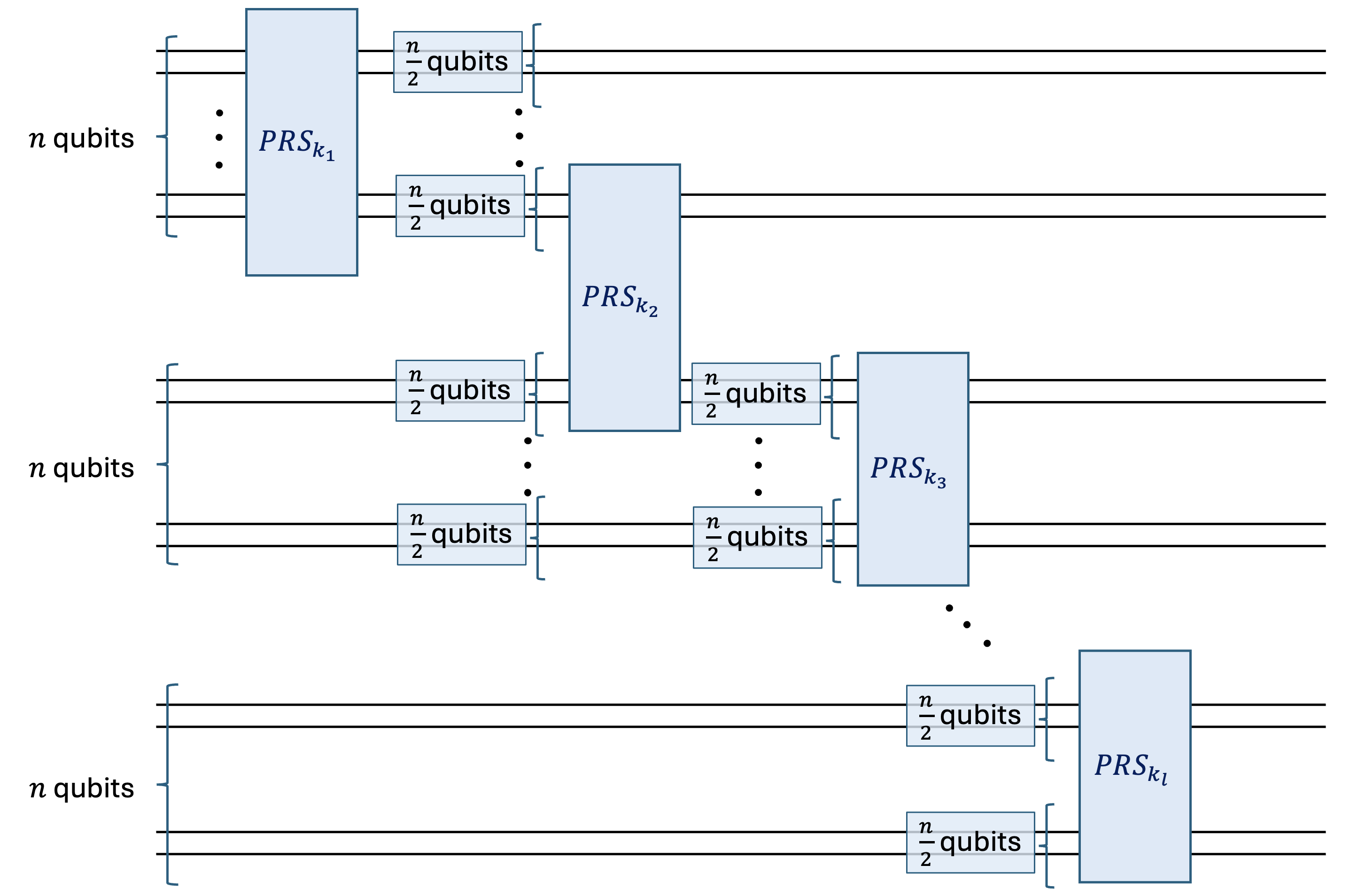}
    \captionof{figure}{\label{fig:construction3}Stairs PRS expansion}
\end{center}

\begin{itemize}
    \item In this construction we are able to expand the length from $n$ to $\frac{n}{2}(\ell+1)$ for any $\ell \in \N$ at the price of $\ell$ different keys.

    \item Note that for the case of $\ell = 3$ we get as output a PRS of size $2n$ as in \hyperref[fig:construction2]{Construction 2}. The key difference is that \hyperref[fig:construction2]{Construction 2} is more efficient, as two PRS blocks operate simultaneously at the start, whereas in this case, they act sequentially. The advantage of this construction, however, is its greater flexibility in adjusting the output length.
\end{itemize}

\subsection{Generalization Condition}

Finally we propose a condition on a PRS that preserves the “good” properties of the PRS generator even when applying it on a
general state. We conjecture that this condition is sufficient to make constructions \hyperref[construction1]{1} (when $i = \frac{n}{2}$), \hyperref[construction2]{2} and \hyperref[construction3]{3} valid length expansions. Specifically, \textit{any} PRS that meets this condition can use these constructions to be expanded.

\begin{definition}[Generalization Condition] \label{condition}
Let $\mathcal{H}$ be a Hilbert space on $n$ qubits, and $\mathcal{K}$ the key space. We say that a PRS satisfy the Generalization Condition if:
    \begin{align}
        \forall x\in \{0,1\}^n \;\; \exists U_x &\in U(\mathcal{H}) \; s.t: \; \forall k \in \mathcal{K}: \nonumber\\
        PRS_k \ket{x} &= U_x \cdot PRS_k \ket{0} \label{eq:cond-1}
    \end{align}
    And $U_x$ is an efficient unitary s.t: 
    \begin{equation}\label{eq:cond-2}
        \exists V,W \in U(\mathcal{H}) \;\;s.t.\;\; \forall y \in \{0,1\}^n :
        \sum_x  \ket{x} \otimes U_x^T \ket{y} = V \ket{y} \otimes W \ket{y}
    \end{equation}
\end{definition}

The motivation for this stems from one of the goals of the purification technique: to obtain a purifying state that resembles the maximally entangled state. The condition helps align the appropriate parts of the purifying register, ensuring that the resulting state resembles the maximally entangled state.

In more detail, the condition~\eqref{eq:cond-1} guarantees that the PRS creation algorithm, when applied on an arbitrary \emph{computational basis} state $\ket{x}$, produces the same pseudorandom state as when applied on the state $\ket{0}$, up to an $x$-dependent unitary $U_x$. Note that crucially, $U_x$ does \emph{not} depend on the key $k$, as otherwise the condition would be vacuous. The condition is natural and could arguably be used to define $PRS_k$ on an input $\ket{x}$; but not all known PRS generation algorithms may lend themselves to it --- for example, we were unable to find an algorithm for the substet PRS~\cite{giurgica2023pseudorandomness,jeronimo2023subset} that would satisfy the condition. 

The second part of the condition,~\eqref{eq:cond-2}, is a technical condition that allows the proof technique of Theorem~\ref{thm1}, based on the ``purification technique'', to go through.

\begin{theorem} \label{BP prs satisfy condition}
    The binary-phase PRS satisfies the Generalization Condition.
\end{theorem}

\begin{proof} Let $x \in \{0,1\}^n$. 
On one hand the algorithm for the binary-phase in~\cite{brakerski2019pseudo} gives, when applied to a general computational basis state $\ket{x}$, 
\begin{equation}\label{PRS_k on x}
    PRS_k\ket{x} = \frac{1}{\sqrt{N}}\sum_{y\in \{0,1\}^n} (-1)^{f_k(y) + x\cdot y}\ket{y}\;.
\end{equation}
Define
\begin{equation*}
    U_x \deff \sum_{y\in \{0,1\}^n} (-1)^{x\cdot y}\ketbra{y}{y},\;\;\; 
    V = H, \;\;\; W = \I\;.
\end{equation*}
Then $\forall y\in \{0,1\}^n$:
\begin{align*}
    \sum_x  \ket{x} \otimes U_x^T \ket{y} &= \sum_x  \ket{x} \otimes \sum_{y'\in \{0,1\}^n} (-1)^{x\cdot y'}\ketbra{y'}{y'} \ket{y}\\
    &= \sum_x  \ket{x} \otimes (-1)^{x\cdot y} \ket{y}\\
    &= \sum_x (-1)^{x\cdot y} \ket{x} \otimes \ket{y}\\
    &= H \ket{y} \otimes \I \ket{y}\;.
\end{align*}
On the other hand we get:
\begin{align*}
    U_x\cdot PRS_k \ket{0} &= U_x \frac{1}{\sqrt{N}}\sum_{y\in \{0,1\}^n} (-1)^{f_k(y) }\ket{y}\\
    &= \sum_{y'\in \{0,1\}^n} (-1)^{x\cdot y'}\ketbra{y'}{y'} \frac{1}{\sqrt{N}}\sum_{y\in \{0,1\}^n} (-1)^{f_k(y)}\ket{y}\\
    &= \frac{1}{\sqrt{N}} \sum_{y'\in \{0,1\}^n}\sum_{y\in \{0,1\}^n} (-1)^{f_k(y) + x\cdot y'}\ketbra{y'}{y'}\ket{y}\\
    &= \frac{1}{\sqrt{N}} \sum_{y\in \{0,1\}^n} (-1)^{f_k(y) + x\cdot y}\ket{y}\\
    &\underset{\eqref{PRS_k on x}}{=} PRS_k\ket{x} \;\;\;.
\end{align*}
\end{proof}

\begin{theorem} \label{GP prs satisfy condition}
    The general-phase PRS satisfy the Generalization Condition.
\end{theorem}

\begin{proof} Let $x \in \{0,1\}^n$. 
On one hand by the algorithm given to the binary-phase at \cite{ji2018pseudorandom}:
\begin{equation}\label{GP PRS_k on x}
    PRS_k\ket{x} = \frac{1}{\sqrt{N}}\sum_{y\in \{0,1\}^n} \omega_N^{f_k(y) + x\cdot y}\ket{y}
\end{equation}

Define
\begin{equation*}
    U_x \deff \sum_{y\in \{0,1\}^n} \omega_N^{x\cdot y}\ketbra{y}{y},\;\;\; 
    V = QFT, \;\;\; W = \I
\end{equation*}
The rest is identical to the proof of Theorem \ref{BP prs satisfy condition}. 

\end{proof}

We conjecture the following. 

\begin{conjecture}\label{conj:1 can be generlized}
    Let $\{C_k\}_k$ be a PRS that satisfy the Generalization Condition and $i(n) = \frac{n}{2}$ the number of added qubits.
    Then \hyperref[fig:construction1]{Construction 1} is a valid length expansion for $\{C_k\}_k$.
\end{conjecture}
\begin{conjecture}\label{conj:2 can be generlized}
    Let $\{C_k\}_k$ be a PRS that satisfy the Generalization Condition.
    Then \hyperref[fig:construction2]{Constructions 2} is a valid length expansion for $\{C_k\}_k$.
\end{conjecture}
\begin{conjecture}\label{conj:3 can be generlized}
    Let $\{C_k\}_k$ be a PRS that satisfy the Generalization Condition.
    Then \hyperref[fig:construction3]{Constructions 3} is a valid length expansion for $\{C_k\}_k$.
\end{conjecture}

%We showed at Theorem \ref{BP prs satisfy condition}, \ref{GP prs satisfy condition} that the binary-phase (as well as the general-phase) PRS  satisfy the \hyperref[condition]{Generalization Condition} and therefore: 
%\begin{itemize}
%    \item We can use \hyperref[fig:construction1]{Construction 1} to add $\frac{n}{2}$ qubits for the binary-phase PRS and general-phase PRS (which aligns with Theorem \ref{thm1}).
%
%    \item Constructions \hyperref[construction2]{2} and \hyperref[construction3]{3}  should be valid length expansion for these PRSs.
%\end{itemize}

%--- Section ---%
\section{Proof of Theorem~\ref{thm1}}
\label{sec:proof1}

In this section we prove Theorem~\ref{thm1}. We start in section~\ref{sec:counting} with some preliminary counting lemmas that will be used in the proof. The main argument is given in Section~\ref{sec:proof-thm1}.

\subsection{Counting lemmas}
\label{sec:counting}

\begin{definition}\label{def:dist-unique}
Let $S$ be a finite set and $t\geq 1$ an integer. For $\ba = (a_1,\ldots,a_t) \in S^t$ we let $\unique(\ba)$  be the set of indices of unique elements among $(a_1,\ldots,a_t)$, i.e.\
\[ \unique(\ba) \,=\, \unique(a_1,\ldots,a_t)\,=\ \big\{ i\in[t]\ :\ a_i\neq a_j\; \forall j\neq i\big\}\;.\]
We also let $\Dist(S;t)$ be the set of $t$-tuples of distinct elements from $S$, i.e.\
\[ \Dist(S;t) \,=\, \big\{ \ba=(a_1,\ldots,a_t)\in S^t\ :\ |\unique(\ba)|=t\big\}\;.\]
When $S=\{0,1\}^n$ for some $n$ we sometimes write $\Dist(n;t)$ for $\Dist(\{0,1\}^n;t)$. 
\end{definition}

\begin{lemma}\label{lem:dist}
For any $n,t\geq 1$, 
\[\big|\Dist(n;t)\big|\geq 2^{nt}\Big(1-\frac{t^2}{2^n}\Big)\;.\]
\end{lemma}

\begin{proof}
We have
\begin{align*}
|\Dist(n;t)| &= 2^n(2^n-1)\cdots (2^n-(t-1))\\
&= 2^{nt}\Big(1-\frac{1}{2^n}\Big)\cdots\Big(1-\frac{t-1}{2^n}\Big)\\
&\geq 2^{nt}\Big(1-\frac{1+\cdots+t-1}{2^n}\Big)\;.
\end{align*}
\end{proof}

\begin{lemma}\label{lem:rpi2}
Let $S$ be a finite set, $t\geq 1$ an integer and $(a_1,\ldots,a_t) \in S^t$. Let $k=|\{a_1,\ldots,a_t\}|$ be the number of distinct elements in the tuple. Let 
\begin{equation}\label{eq:def-perm}
 \ket{\{a_1,\ldots,a_t\}} \,=\, \frac{1}{\sqrt{t!}}\sum_{\pi\in \mathfrak{S}_t} R_\pi \ket{a_1,\ldots,a_t}\;.
\end{equation}
Then 
\[ \big\|\ket{\{a_1,\ldots,a_t\}} \big\|^2 \,\leq\, (t-k+1)!\;.\]
\end{lemma}

\begin{proof}
For $i=1,\ldots,k$ let $a'_i$ be the $i$-th element of $\{a_1,\ldots,a_t\}$ in lexicographic order, and define $S_i = \{j:\ a_j=a'_i\}$. For permutation $\pi,\pi' \in \mathfrak{S}_t$, $R_\pi \ket{a_1,\ldots,a_t} = R_{\pi'} \ket{a_1,\ldots,a_t}$ if $\pi^{-1}\circ \pi'$ decomposes as a product of permutations, each of which is identity outside of a set $S_i$ for some $i$. Write $\pi\sim \pi'$ if this condition is satisfied. For any permutation $\pi$ there are at most $|S_1|! \cdots |S_k|!$ permutations $\pi'$ such that $\pi\sim\pi'$. This is an equivalence relation, hence $\mathfrak{S}_k$ is a disjoint union of equivalence classes, each of which has at most $|S_1|! \cdots |S_k|!\leq (t-k+1)!$ elements. Furthermore, if $\pi$, $\pi'$ are in \emph{different} equivalence classes then $R_{\pi} \ket{a_1,\ldots,a_t}$ and $R_{\pi'} \ket{a_1,\ldots,a_t}$ are orthogonal. Letting $E_1,\ldots,E_r$ be the equivalence classes, 
\begin{align*}
\big\|\ket{\{a_1,\ldots,a_t\}} \big\|^2 &= \frac{1}{t!} \sum_j \Big\| \sum_{\pi\in E_j} R_\pi \ket{a_1,\ldots,a_t}\Big\|^2\\
&= \frac{1}{t!} \sum_j |E_j|^2 \\
&\leq (t-k+1)!\frac{1}{t!} \sum_j |E_j|\\
&= (t-k+1)!\;.
\end{align*}
\end{proof}

%For the next lemma, given sets $T,T'$ we let $\surj(T',T)$ denote all surjections from $T'\to T$. We also use the shorthand $\surj(t',t)$ for $\surj([t'],[t])$. 

\subsection{Proof of Theorem~\ref{thm1}}\label{sec:proof-thm1}

We now prove Theorem~\ref{thm1}. The main idea in the proof is to use the purification technique introduced in~\cite{ma2024construct}; however, due to the nature of our construction, where the same cryptographic key is used for both building blocks, there are some additional technical challenges to the argument. 

First note that given $k \in \mathcal{K}$ the algorithm to create $\ket{\psi_k}$ is efficient as it calls twice $PRS_k$, which is an efficient quantum algorithm by the definition of PRS. In addition our algorithm use one layer of Hadamard gates at the end. Therefore overall this algorithm is an efficient quantum algorithm.  

Hence, in order to show that $\{\ket{\psi_k}\}_k$ is a PRS we need to show that:
\[\forall t = \poly(\lambda)\ ,\quad \E_{k}[\ketbra{\psi_k}{\psi_k}^{\otimes t}] \,\overset{c}{\approx}\, \E_{\mu\leftarrow Haar}[\ketbra{\mu}{\mu}^{\otimes t}]\;.\]
In the binary-phase construction $k$ determines a quantum secure PRF $f'$.
We denote $\ket{\psi_k}$ as $\ket{\psi_{f'}}$ i.e: 
$\E_{k}[\ketbra{\psi_k}{\psi_k}^{\otimes t}] = \E_{f'}[\ketbra{\psi_{f'}}{\psi_{f'}}^{\otimes t}]$.
We can replace the quantum secure PRF $f'$ with a random function $f$ to get a state $\ket{\psi_f}$. It is not hard to show that
\begin{equation}
    \E_{k}[\ketbra{\psi_k}{\psi_k}^{\otimes t}] = \E_{f'}[\ketbra{\psi_{f'}}{\psi_{f'}}^{\otimes t}] \overset{c}{\approx} 
    \E_{f\in \{0,1\}^N}[\ketbra{\psi_{f}}{\psi_{f}}^{\otimes t}]\;,
\end{equation}
where the computational indistinguishability $\overset{c}{\approx}$ follows from quantum security of the PRF. 
Therefore showing $\E_{f\in \{0,1\}^N}[\ketbra{\psi_{f}}{\psi_{f}}^{\otimes t}] \overset{c}{\approx} \E_{\mu\leftarrow Haar}[\ketbra{\mu}{\mu}^{\otimes t}]$ will finish the proof.
First note that, setting $N=2^n$,
\begin{equation}
\label{exp3}
    \E_{f\in \{0,1\}^N}[\ketbra{\psi_{f}}{\psi_{f}}^{\otimes t}] = \frac{1}{\sqrt{2^N}}\Tr_E\Big[\sum_f \ket{\psi_f}^{\otimes t} \otimes \ket{f}_E \Big]\;.
\end{equation}
Here we mean the partial trace on the density matrix of this pure state as in Definition \ref{def1}, namely
\begin{equation*}
    \Tr_E\Big[\sum_f \ket{\psi_f}^{\otimes t} \otimes \ket{f}_E \Big] \deff 
    \Tr_E\Big[\Big(\sum_f \ket{\psi_f}^{\otimes t} \otimes \ket{f}_E \Big)\cdot \Big(\sum_f \bra{\psi_f}^{\otimes t} \otimes \bra{f}_E \Big)\Big]\;.
\end{equation*}
Now let us analyze $\ket{\psi_f}$ based on the fact that $PRS_k$ is the binary-phase PRS, implemented as described in Definition~\ref{def:prs-algo}.
\begin{align}
\ket{0}^{\otimes n+i} \xrightarrow{PRS_k\otimes \I^{\otimes i}}\quad&
         \frac{1}{\sqrt{2^n}} \sum_{x\in \{0,1\}^n} (-1)^{f(x)} \ket{x} \otimes \ket{0}^{\otimes i} \notag \\ = \quad &
         \frac{1}{\sqrt{2^n}} \sum_{x'\in \{0,1\}^i} \sum_{x''\in \{0,1\}^{n-i}} (-1)^{f(x'x'')} \ket{x'}\ket{x''} \otimes \ket{0}^{\otimes i}\label{eq:xpp-0} \\
         \xrightarrow{\I^{\otimes i} \otimes PRS_k}\quad & 
         \frac{1}{2^n} \sum_{x'\in \{0,1\}^i} \sum_{x''\in \{0,1\}^{n-i}} \sum_{y \in \{0,1\}^n} (-1)^{f(x'x'') + y\cdot (x''0^i) + f(y)} \ket{x'} \ket{y}\label{eq:ypp-0}\\
         \xrightarrow{H^{\otimes (n+i)}}\quad&
         \frac{1}{2^n} \sum_{x'\in \{0,1\}^i} \sum_{x''\in \{0,1\}^{n-i}} \sum_{y \in \{0,1\}^n} (-1)^{f(x'x'') + y\cdot (x''0^i) + f(y)} H\left(\ket{x'}\ket{y}\right)\;.\notag
\end{align}
The state in the partial trace on the right-hand side of Eq. \eqref{exp3} is thus
\begin{align}
        &\sum_f \ket{\psi_f}^{\otimes t} \otimes \ket{f}_E \notag\\
        &=
        \sum_f \frac{1}{2^{tn}} \sum_{\substack{x'_1,...,x'_t \in \{0,1\}^i\\x''_1,...,x''_t \in \{0,1\}^{n-i}\\y_1,...,y_t \in \{0,1\}^n}} (-1)^{\sum_{j=1}^t f(x'_j x''_j) + y_j\cdot (x''_j 0^i) + f(y_j)} H\big(\ket{x'_1}\ket{y_1} \otimes ... \otimes \ket{x'_t}\ket{y_t}\big)
        \otimes \ket{f}_E\;.    \label{eq13}
\end{align}
Let $\mathbf{x'} \deff$ $(x'_1,...,x'_t)$ (i.e.\ $\mathbf{x'}$ is a vector of $t$ elements s.t each element is a bit string of size $i$), $\mathbf{x''} \deff$ $(x''_1,...,x''_t)$ (i.e.\ $\mathbf{x''}$ is a vector of $t$ elements s.t each element is a bit string of size $n-i$), and $\mathbf{y} \deff$ $(y_1,...,y_t)$ (i.e.\ $\mathbf{y}$ is a vector of $t$ elements s.t each element is a bit string of size $n$).
For simplicity denote
\begin{equation}
\label{phi}
    \ket{\phi_{\mathbf{x'},\mathbf{y}}} \,\deff\, H\big(\ket{x'_1}\ket{y_1} \otimes ... \otimes \ket{x'_t}\ket{y_t}\big)\;.
\end{equation}

\begin{claim}
\label{clm5} 
\[ \E_{f\in \{0,1\}^N}[\ketbra{\psi_{f}}{\psi_{f}}^{\otimes t}] \,=\, \Tr_E\big[ \ket{\psi_{all}} \big]\;,\]
where $\ket{\psi_{all}}$ is the normalized state 
\[ \ket{\psi_{all}}\,=\,\frac{1}{\sqrt{2^{(n+i)t}}} \sum_{\mathbf{x', y}} \ket{\phi_{\mathbf{x'},\mathbf{y}}}\ket{\phi'_{\bxp,\by}}_E\;,\]
with  the state     $\ket{\phi_{\mathbf{x'},\mathbf{y}}}$ defined in~\eqref{phi} and 
				\[ \ket{\phi'_{\bxp,\by}} \,=\, \frac{1}{\sqrt{2^{(n-i)t}}}\sum_{\bxpp} (-1)^{\sum_{j=1}^t y_j\cdot (x''_j 0^i)} \ket{e_{x'_1x''_1}\oplus e_{y_1}\oplus \cdots \oplus e_{x'_tx''_t}\oplus e_{y_t}}\;.\]
\end{claim}

\begin{proof}
Using~\eqref{eq13} and then~\eqref{phi},
\begin{align}
        \frac{1}{\sqrt{2^N}} \sum_f \ket{\psi_f}^{\otimes t} \otimes \ket{f}_E
&=    \frac{1}{\sqrt{2^N}} \sum_f \frac{1}{2^{tn}} \sum_{\substack{x'_1,...,x'_t \in \{0,1\}^i\\ x''_1,...,x''_t \in \{0,1\}^{n-i}\\ y_1,...,y_t \in \{0,1\}^n}} (-1)^{\sum_{j=1}^t f(x'_j x''_j) + y_j\cdot (x''_j 0^i) + f(y_j)} \ket{\phi_{\mathbf{x'},\mathbf{y}}}
    \otimes \ket{f}_E \notag\\
&=    \frac{1}{ 2^{tn}} \sum_{\mathbf{x', x'', y}}  \ket{\phi_{\mathbf{x'},\mathbf{y}}}
    \otimes \frac{1}{\sqrt{2^N}} \sum_f (-1)^{\sum_{j=1}^t f(x'_j x''_j) + y_j\cdot (x''_j 0^i) + f(y_j)} \ket{f}_E \notag\\
&= \frac{1}{2^{tn}} \sum_{\mathbf{x', x'', y}} \ket{\phi_{\mathbf{x'},\mathbf{y}}}
    \otimes (-1)^{\sum_{j=1}^t y_j\cdot (x''_j 0^i) } \frac{1}{\sqrt{2^N}}\sum_f (-1)^{\sum_{j=1}^t f(x'_j x''_j) + f(y_j)} \ket{f}_E\;.
  \label{exp1}
\end{align}
%Let 
%\[ \ket{\phi''_{\bxp,\by}} \,=\,\frac{1}{\sqrt{2^{(n-i)t}}} \sum_{\bx''} (-1)^{\sum_{j=1}^t y_j\cdot (x''_j 0^i) } \frac{1}{\sqrt{2^N}}\sum_f (-1)^{\sum_{j=1}^t f(x'_j x''_j) + f(y_j)} \ket{f}_E\;.\]
%Then 
%\begin{align}
%\big\| \ket{\phi''_{\bxp,\by}} \big\|^2&=\frac{1}{2^{(n-i)t}}\frac{1}{2^N} \sum_f \Big|\sum_{\bx''} (-1)^{\sum_{j=1}^t y_j\cdot (x''_j 0^i) + f(x'_j x''_j) + f(y_j)}\Big|^2 \notag\\
%%&=\frac{1}{2^N} \sum_f  \prod_j \Big| \sum_{x''_j} (-1)^{ y_j\cdot (x''_j 0^i) + f(x'_j x''_j)}\Big|^2\\
%&=\frac{1}{2^{(n-i)t}}  \Big( \sum_{\bxpp,\mathbf{z''}} (-1)^{ \sum_{j=1}^t y_j\cdot ((x''_j 0^i)+(z''_j 0^i))} \Big(\frac{1}{2^N} \sum_{f} (-1)^{\sum_{j=1}^t f(x'_j x''_j)+f(x'_j z''_j)}\Big)\Big)\notag\\
%&= \frac{1}{2^{(n-i)t}}\Big( \sum_{\substack{\bxpp,\mathbf{z''}\\ \oplus_j (e_{x'_jx''_j}\oplus e_{x'_jz''_j})=0}} (-1)^{ \sum_{j=1}^t y_j\cdot ((x''_j 0^i)+(z''_j 0^i))}\Big)\notag\\
%&= 1\;.\label{eq:phipp-norm}
%\end{align}
%Here the last line is justified because the parity condition determines exactly one variable $x''_j$ for each set of indices associated with the same $x'_j$. 

The function $f$ can be represented as a long vector of $N=2^n$ entries such that the $i$-th entry 
represents $f(i) \in \{0,1\}$ and can be accessed by the inner product: $f\cdot e_i = f(i)$. In this notation, we can rewrite 
\begin{equation*}
  (-1)^{\sum_{j=1}^t f(x'_ix''_j) + f(y_j)} \,=\, (-1)^{f\cdot(e_{x'_1x''_1}\oplus e_{y_1}...\oplus e_{x'_tx''_t}\oplus e_{y_t})}\;.
\end{equation*}
Plugging this in to Eq. \eqref{exp1} we get:
\begin{equation*}
         \frac{1}{2^{tn}} \sum_{\mathbf{x', x'', y}} \ket{\phi_{\mathbf{x'},\mathbf{y}}}
        \otimes \underbrace{(-1)^{\sum_{j=1}^t y_j\cdot (x''_j 0^i) } \frac{1}{\sqrt{2^N}} \sum_f (-1)^{f\cdot (e_{x'_1x''_1}\oplus e_{y_1}...\oplus e_{x'_tx''_t}\oplus e_{y_t})} \ket{f}_E}_{\text{the purifying register}}\;.
\end{equation*}
Recall that according to Claim \ref{clm1} we can apply any unitary or isometry to the 
purifying register and still get a purification of the original state.
Also,
\[H\ket{e_{x'_1x''_1}\oplus e_{y_1}...\oplus e_{x'_tx''_t}\oplus e_{y_t}} = \frac{1}{\sqrt{2^N}}\sum_f (-1)^{f\cdot (e_{x'_1x''_1}\oplus e_{y_1}...\oplus e_{x'_tx''_t}\oplus e_{y_t})} \ket{f}\;.\]
So by applying $H$ on the purifying register we get another purification for $\E_{f\in \{0,1\}^N}[\ketbra{\psi_{f}}{\psi_{f}}^{\otimes t}]$,
\begin{equation}
\label{exp2}
         \frac{1}{2^{tn}} \sum_{\mathbf{x', x'', y}} \ket{\phi_{\mathbf{x'},\mathbf{y}}}
        \otimes (-1)^{\sum_{j=1}^t y_j\cdot (x''_j 0^i)} \ket{e_{x'_1x''_1}\oplus e_{y_1}...\oplus e_{x'_tx''_t}\oplus e_{y_t}} = \ket{\psi_{\textit{all}}}\;.
\end{equation}
This shows the claim.%; where the normalization claim follows from~\eqref{eq:phipp-norm} and the fact that $\ket{\phi'_{\bxp,\by}}$ is obtained from $\ket{\phi''_{\bxp,\by}}$ by applying an isometry. 
\end{proof}

\begin{claim}\label{claim:5b}
For every $\bxp,\by$ let $\ket{\phi_{\bxp,\by}}$ be as in~\eqref{phi} and let
\begin{align*}
 \ket{\tilde{\phi}'_{\bxp,\by}} &= \frac{1}{\sqrt{2^{(n-i)t}}}\sum_{\bxpp\in \Dist(n-i;t)} (-1)^{\sum_{j=1}^t y_j\cdot (x''_j 0^i)} \ket{e_{x'_1x''_1}\oplus e_{y_1}\oplus \cdots \oplus e_{x'_tx''_t}\oplus e_{y_t}}\;,\\
\ket{{\psi}_{all}'} &= \frac{1}{\sqrt{2^{(n+i)t}}}\sum_{\substack{\bxp,\by\\ \by_>\in \Dist(n-i;t)}} \ket{\phi_{\bxp,\by}} \otimes \ket{\tilde{\phi}'_{\bxp,\by}} \;,
\end{align*}
where $\by_> = (y_{1,>},\ldots,y_{j,>})$ and for each $j$, $y_{j,>}$ contains the last $(n-i)$ bits of $y_j$. 
Then 
\[ \big\| \ket{\psi_{all}} - \ket{{\psi}_{all}'} \big\|^2\,=\,\negl(\lambda)\;.\]
%\[ \frac{1}{2^{(n+i)t}}\sum_{\bxp,\by} \big\|\ket{\phi'_{\bxp,\by}}-\ket{\tilde{\phi}'_{\bxp,\by} }\big\|^2\,=\, \negl(\lambda)\;.\]
\end{claim}

\begin{proof}
If, in the procedure that leads to the preparation of $\ket{\psi_{all}}$, at step~\eqref{eq:xpp-0} (after application of the first $PRS_k$), we perform a binary projective measurement on the register containing $x''$ to determine if $x''\in \Dist(n-i;t)$ or not, we will obtain the first outcome with probability $|\Dist(n-i;t)|/2^{(n-i)t} = 1-\negl(\lambda)$ by Lemma~\ref{lem:dist}, and the second outcome with probability $\negl(\lambda)$. The same reasoning applies if the measurement is applied at step~\eqref{eq:ypp-0} (after application of the second $PRS_k$), on the registers containing $\by_>$. 

Furthermore, conditioned on both measurements being performed and the first outcome being obtained both times (but without renormalizing), following all remaining steps we see that the state 
$\ket{{\psi}_{all}'}$
will be prepared. Since all steps of the preparation procedure are unitary, and unitaries contract the trace distance, the claim follows from two applications of the gentle measurement lemma, Lemma~\ref{lem:gentle}. 
\end{proof}

Note that $e_{x'_1x''_1}\oplus e_{y_1}...\oplus e_{x'_tx''_t}\oplus e_{y_t} \in \{0,1\}^N$ is a long vector of $N$ entries.
%For any $x\in \{0,1\}^d, d\in \N$ denote $W(\ket{x}) \deff \sum_{i=1}^d x_i$
%Note that $W(\ket{e_{x'_1x''_1}\oplus e_{y_1}...\oplus e_{x'_tx''_t}\oplus e_{y_t}}) = 2t \iff (x'_1x''_1, y_1, ... , x'_tx''_t, y_t) \in \Dist$, where $\Dist$ is the set of all distinct elements of tuples of size $2t$, and each element of the 
%tuples is a bit string of size $n$. 
Recall the notation introduced in Definition~\ref{def:dist-unique} and define the set 
\begin{equation}\label{Dist def}
    \Dist \deff \big\{(\bxp,\bxpp,\byl,\byr)\in\{0,1\}^{it}\times\{0,1\}^{(n-i)t}\times\{0,1\}^{it}\times \{0,1\}^{(n-i)t}  \ :\ \bxpp\circ\byr\in \Dist(n-i;2t)\big\}\;.
\end{equation} 
In other words (and slightly rearranging indices),  $\Dist =\{0,1\}^{2it} \times \Dist(n-i,2t)$. We will sometimes write that a tuple $(\bxp,\bxpp,\by)\in \Dist$ if $\by=(\byl,\byr)\in\{0,1\}^{it}\times\{0,1\}^{(n-i)t}$. By Lemma~\ref{lem:dist},
\begin{equation}\label{eq:dist-size}
|\Dist| \,=\, 2^{2it}\big|\Dist(n-i;2t)\big| \,\geq\, N^{2t}\Big(1-\frac{O(t^2)}{N}\Big)\;.
\end{equation}

Before proceeding with the proof we first establish a slightly more general lemma, that will be used twice. It may be better for the reader to skip the lemma at first, and return to it the first time it is used.

\begin{lemma}\label{lem:helper}
For $T,S\subseteq\{0,1\}^n$ and $\by\in(\{0,1\}^{n})^{t}$ let $\ket{u_{T,\by}}$ and $\ket{v_S}$ be such that
\begin{enumerate}
\item[(i)] For any $T$ and $\by$, $\|\ket{u_{T,\by}}\|^2\leq t^{t-|T|+1}$;
\item[(ii)] For any $S$, $\|\ket{v_S}\|^2 \leq 1$; 
\item[(iii)] For $T,T'$ such that $T\neq T'$ and $T\cap \{y_1,\ldots,y_t\}=T'\cap\{y_1,\ldots,y_t\}=\emptyset$, $\ket{u_{T,\by}}$ and $\ket{u_{T',\by}}$ are orthogonal.
\end{enumerate}
Then 
\begin{align*}
  \sum_{\substack{\bxp, \by\\\by\in \Dist(n;t)}} \Big\| \sum_{\substack{\mathbf{x''} \\ (\bxp,\bxpp,\by)\notin\Dist}} (-1)^{\sum_{j=1}^t y'_j\cdot x''_j} \ket{u_{\{x_1,\ldots,x_t\}\backslash\{y_1,\ldots,y_t\},\by}}\ket{v_{\{x_1,\ldots,x_t\}\cap\{y_1,\ldots,y_t\}}} \Big\|^2 \,=\, N^{2t} \negl(\lambda)\;.
\end{align*}
\end{lemma}

\begin{proof}
We first write 
\begin{align}
 \sum_{\substack{\bxp,\by\\ \by\in \Dist(n;t)}}& \Big\|  \sum_{\substack{\mathbf{x''} \\ (\bxp,\bxpp,\by)\notin\Dist}} (-1)^{\sum_{j=1}^t y'_j\cdot x''_j}  \ket{u_{\{x_1,\ldots,x_t\}\backslash\{y_1,\ldots,y_t\},\by}}\ket{v_{\{x_1,\ldots,x_t\}\cap\{y_1,\ldots,y_t\}}} \Big\|^2 \notag\\
 % &= \sum_{\bxp,\by} \Big\| \sum_{\bxpp: (\bxp,\bxpp,\by)\notin \Dist} \otimes (-1)^{\sum_{j=1}^t y_j\cdot (x''_j 0^i)} \ket{\{x'_1x''_1,\ldots,x'_tx''_t\}} \Big\|^2 \\
  &= \sum_{\substack{\bxp,\by\\\by\in\Dist(n;t)}} \Big\|  \sum_{S\subseteq[t]}\sum_{\substack{T\subseteq \{0,1\}^{n}\\  T\cap\{y_{1},\ldots,y_{t}\}=\emptyset }} \sum_{\substack{\pi\in\surj(t,T\cup y_S)}} (-1)^{\sum_{j=1}^t y'_j \cdot \pi(j)'' } \ket{u_{T,\by}}\ket{v_{\by_S}} \Big\|^2 \;. \label{eq:fermi1-1}
\end{align}
Here, we introduced sets $S\subseteq [t]$ and $T=\{x_1,\ldots,x_t\}\backslash\{y_1,\ldots,y_t\}$ from which the ordered $t$-tuple $\bx$ is uniquely recovered by choosing a $\pi:[t]\to T\cup\{y_j:\ j\in S\}$, required to be a surjection to avoid double-counting. By $\pi(j)''$ we denote the last $(n-i)$ bits of $\pi(j)$. Note that, inside the squared norm, we should include an indicator of the event that $\pi(j)'=x_j'$ for all $j\in [t]$, where $\pi(j)'$ denotes the first $i$ bits of $\pi(j)$, which we omitted for clarity; and we should remember the requirement that $(\bxp,\bxpp,\by)\notin \Dist$, which here is a requirement on the tuple $(\by,T,\pi)$---this will express itself later on. 

Since the vectors $ \ket{u_{T}}$ are orthogonal for distinct $T$, we get that
\begin{align}
\eqref{eq:fermi1-1}&= \sum_{\substack{\bxp,\by\\\by\in\Dist(n;t)}}  \sum_{\substack{T\subseteq \{0,1\}^{n}\\  T\cap\{y_{1},\ldots,y_{t}\}=\emptyset }} \big\| \ket{u_{T,\by}} \big\|^2 \Big|\sum_{S\subseteq [t]} \sum_{\substack{\pi\in\surj(t,T\cup \by_S) }} (-1)^{\sum_{j=1}^t y'_j \cdot \pi(j)'' } \ket{v_S}\Big|^2\notag\\
&\leq \sum_{\substack{\bxp,\by\\\by\in\Dist(n;t)}} \sum_{\substack{T\subseteq \{0,1\}^{n}\\ |T|< t}} t^{t-|T|+1} \Big|\sum_{S\subseteq [t]} \sum_{\substack{\pi\in\surj(t,T\cup \by_S) }} (-1)^{\sum_{j=1}^t y'_j \cdot {\pi(j)}'' }\ket{v_S}\Big|^2\notag\\
%&= \sum_{\substack{\by\\\byp\in\Dist(n;t)}} \sum_{\substack{T\subseteq \{0,1\}^{n}}} t^{t-|T|+1} \Big|\sum_{S\subseteq [t]} \sum_{\substack{\pi\in \surj(t, T\cup y_S )}} (-1)^{\sum_{j=1}^t y'_j \cdot {\pi(j)}'' } \ket{v_S}\Big|^2\notag\\
%&\leq \sum_{\substack{T\subseteq \{0,1\}^{n}\\ |T|< t}}\big(t-|T|+1\big)! \sum_{\bxp,\by} \sum_{S,S'\subseteq [t]}   \sum_{\substack{\pi\in \surj(t,T\cup y_S)\\ \pi'\in \surj(t,T\cup y_{S'})}} (-1)^{\sum_{j=1}^t y'_j \cdot (\pi(j)''+\pi'(j)'') } \notag\\
&\leq  \sum_{\substack{T\subseteq \{0,1\}^{n}}} t^{t-|T|+1}  \sum_{S,S'\subseteq [t]}   \sum_{\bxp, \by_{S\cup S'}}  \sum_{\substack{\pi\in \surj(t,T\cup y_S)\\ \pi'\in \surj(t,T\cup y_{S'})}} \sum_{\by_{[t]\backslash (S\cup S')} } (-1)^{\sum_{j=1}^t y'_j \cdot (\pi(j)''+\pi'(j)'') } \bra{v_S} v_{S'}\rangle \notag\\
&\leq \sum_{\substack{T\subseteq \{0,1\}^{n}}} t^{t-|T|+1}  \sum_{S,S'\subseteq [t]}   \sum_{\by_{S\cup S'}}  \sum_{\substack{\pi\in \surj(t,T\cup y_S)\\ \pi'\in \surj(t,T\cup y_{S'})}}\Big| \sum_{\by_{[t]\backslash (S\cup S')} } (-1)^{\sum_{j=1}^t y'_j \cdot (\pi(j)''+\pi'(j)'') }\Big| \notag\\
&\leq 2^{it}  \sum_{\substack{T'\subseteq \{0,1\}^{n-i}\\ |T'|<t}}t^{t-|T'|+1} \sum_{S,S'\subseteq [t]}  \sum_{\by_{S\cup S'}} \sum_{\substack{\pi\in \surj(t,T'\cup y_{S,>})\\ \pi'\in \surj(t,T'\cup y_{S',>})}} \prod_{j\notin S\cup S'}  \Big|\sum_{y_j} (-1)^{y'_j \cdot (\pi(j) +\pi'(j))} \Big|\;.
    \label{eq:fermi1-2}
\end{align}
Here in the second line we used assumption (i) to bound $\|\ket{u_{T,\by}}\|^2$ by $t^{t-|T|+1}$. Further, we removed the restriction on $T\cap\{y_1,\ldots,y_t\}=\emptyset$. 
In the third line we drop the restriction on $\byp\in\Dist(n-i;t)$ and expand the square; we let $y_S$ denote the set $\{y_j:\ j\in S\}$. 
In the fourth line we insert an absolute value and use $|\bra{v_S}v_{S'}\rangle|\leq 1$. At this point we can also remove the summation over $\bxp$ as well as the (implicit) indicator that $\pi(j)'=x_j'$ for all $j\in [t]$, because for any $T$ and $\pi,\pi'$ there is (at most) one $\bxp$ that will satisfy the condition for both $\pi$ and $\pi'$. 
Finally for the last line we use that the summation inside absolute values only depends on the last $(n-i)$ bits of elements of $T$. The factor $2^{it}$ in front accounts for the first $i$ bits of each element of $T$ that were dropped (note $|T|\leq t$). Since $|T'|\leq |T|$ we have $t^{t-|T'|+1}\geq t^{t-|T|+1}$, and counting surjections to $T'\cup y_{S,>}$ instead of surjections to $T\cup y_S$ only contributes more terms (since a surjection to the latter induces a surjection to the former by restriction). Finally, we used the condition that $\bxpp\circ \by_>\in\Dist(n-i;2t)$, which was implicit in the notation up to here, to require $|T'|<t$. 

The sum $\sum_{y_j} (-1)^{y'_j \cdot (\pi(j) +\pi'(j))}$ is $0$ if $\pi(j)\neq \pi'(j)$, and it is $2^{n}$ otherwise. Letting $k=|T|$ and $s=|S\cup S'|$ and continuing, 
\begin{align*}
 \eqref{eq:fermi1-2}&\leq 2^{it} \sum_{k=0}^{t-1} \sum_{s=0}^{2(t-k)} {2^{n-i}\choose k} t^{t-k+1} 2^{(n-i)s} 2^{2s} \cdot t!\ (t+s)^{t-k} \cdot (t+s)^{s} \cdot  2^{n(t-s)}  \notag \\
&\leq  N^{2t}\cdot O(t^3) \Big(\frac{t\cdot 2^4  \cdot (t+s)^3}{2^{n-i}} \Big)^{t-k} \frac{t!}{k!}\\
&= N^{2t} \negl(\lambda)\;.
\end{align*}
Here for the first line we used that $s=|S\cup S'|\leq |S|+|S'|\leq 2(t-k)$, because $|T'\cup y_S|,|T'\cup y_{S'}|\leq t$ as otherwise there is no surjection from $[t]$ to them. Choosing $S,S'$ is the same as choosing $s$ elements (at most $2^{(n-i)s}$ choices) and then choosing which are in $S$ and which in $S'$ (at most $2^s\times 2^s$ choices). To count the number of options for $\pi$, we first choose $k$ unique preimages for the elements of $T'$, leading to at most $t!$ choices, then images for the remaining elements, leading to at most $(t+s)^{t-k}$ choices. To count the number of options for $\pi'$, $(t-s)$ elements are already determined by the condition $\pi'(j)=\pi(j)$ for $j\notin S\cup S'$, and there are at most $(t+s)^{s}$ choices for the others. The second line regroups terms and the last line uses $t!/k! \leq t^{t-k}$, and $t-k>0$ and $n-i=\omega(\log \lambda)$, $t=\poly(\lambda)$ to conclude. 
\end{proof}

\begin{claim}
\label{clm2}
   Let 
    \begin{align*}
       \ket{\psi_{all}} &=   \frac{1}{2^{tn}} \sum_{\mathbf{x', x'', y}} \ket{\phi_{\mathbf{x'},\mathbf{y}}}
            \otimes (-1)^{\sum_{j=1}^t y_j\cdot (x''_j 0^i)} \ket{e_{x'_1x''_1}\oplus e_{y_1}...\oplus e_{x'_tx''_t}\oplus e_{y_t}}\;,\\
       {\ket{\psi_{Dist}}} &=     \frac{1}{\sqrt{|\Dist|}} \sum_{(\bxp, \bxpp,\by)\in \Dist} \ket{\phi_{\mathbf{x'},\mathbf{y}}}\otimes (-1)^{\sum_{j=1}^t y_j\cdot (x''_j 0^i)} \ket{e_{x'_1x''_1}\oplus e_{y_1}...\oplus e_{x'_tx''_t}\oplus e_{y_t}}\;. 
\end{align*}
Then 
\[\ket{\psi_{all}} \,\overset{s}{\approx}\, \ket{\psi_{Dist}}\;.\]
\end{claim}

\begin{proof}
We want to show that $TD\left( \rho_{\textit{all}}, \rho_{Dist} \right) = \negl(\lambda)$, where $\rho_{\textit{all}} = \ketbra{\psi_{\textit{all}}}{\psi_{\textit{all}}}$ and $\rho_{Dist} = \ketbra{\psi_{Dist}}{\psi_{Dist}}$.
According to Fact~\ref{fact1}, 
\begin{equation}\label{eq:clm2-1}
    TD\left( \rho_{\textit{all}}, \rho_{Dist} \right) = \sqrt{1 - |\bra{\psi_{\textit{all}}}\psi_{Dist}\rangle|^2}\;.
\end{equation}
%Let's define some parameters for which to apply Lemma~\ref{lem:1}: 
%\begin{itemize}
%\item $S \deff (\{0,1\}^{it} \cross \{0,1\}^{nt}) \cross \{0,1\}^{(n-i)t}$ 
%\item $x \deff (\mathbf{x',y})$ 
%\item $y \deff \mathbf{x''}$ 
%\item $S' \deff \{(\mathbf{(x',y), x''}) \; | \; \mathbf{x', x'', y}\in Dist \}$
%\item $S'_x \deff \{\mathbf{x''} \; | \; \mathbf{x', x'', y}\in Dist \}$
%\item $\ket{x}$ in the lemma is $\ket{\phi_{\mathbf{x'},\mathbf{y}}}$
%\item $\ket{u_y}$ in the lemma is $(-1)^{\sum_{j=1}^t y_j\cdot (x''_j 0^i)} \ket{e_{x'_1x''_1}\oplus e_{y_1}...\oplus e_{x'_tx''_t}\oplus e_{y_t}}$.
%\end{itemize}
%With these choices, the conditions of the lemma are satisfied: indeed,  $\forall y: \; \norm{\ket{u_y}} = 1$ and $\forall y\in S'_x: \; \ket{u_y}$ are pairwise orthogonal. 

%Note that using notation from the lemma, with our choices we have 
%\begin{equation}\label{eq:clm2-2}
%\ket{\psi_{\textit{all}}} = \frac{1}{2^{tn}} \ket{\psi_S} = \frac{1}{\sqrt{|S|}} \ket{\psi_S}\qquad\text{and}\qquad %\ket{\psi_{Dist}} = \frac{1}{\sqrt{\abs{Dist}}}\ket{\psi_{S'}} = \frac{1}{\sqrt{\abs{S'}}} \ket{\psi_{S'}}\;.
%\end{equation}
Let 
\begin{align*}
 \ket{\tilde{\psi}_{all}} &=   \sum_{\mathbf{x', x'', y}} \ket{\phi_{\mathbf{x'},\mathbf{y}}}
            \otimes (-1)^{\sum_{j=1}^t y_j\cdot (x''_j 0^i)} \ket{e_{x'_1x''_1}\oplus e_{y_1}\oplus\cdots \oplus e_{x'_tx''_t}\oplus e_{y_t}}\;,\\
       {\ket{\tilde{\psi}_{Dist}}} &=    \sum_{(\bxp, \bxpp,\by) \in \Dist} \ket{\phi_{\mathbf{x'},\mathbf{y}}}\otimes (-1)^{\sum_{j=1}^t y_j\cdot (x''_j 0^i)} \ket{e_{x'_1x''_1}\oplus e_{y_1}\oplus \cdots \oplus e_{x'_tx''_t}\oplus e_{y_t}}\;,\\
			 {\ket{\tilde{\psi}_{all}'}} &=    \sum_{\substack{\bx,\by \\ \bxpp, \by_>\in \Dist(n-i;t)}} \ket{\phi_{\mathbf{x'},\mathbf{y}}}\otimes (-1)^{\sum_{j=1}^t y_j\cdot (x''_j 0^i)} \ket{e_{x'_1x''_1}\oplus e_{y_1}\oplus \cdots \oplus e_{x'_tx''_t}\oplus e_{y_t}}\;,
\end{align*}
be unnormalized states. Using Claim~\ref{claim:5b},
%We can evaluate
\begin{align*}
\big\|  \ket{\tilde{\psi}_{all}} -  {\ket{\tilde{\psi}_{all}'}}\big\|^2
%&= \sum_{\substack{\bxp,\by\\\by_{>}\notin \Dist(n-i;t)}} 2^{(n-i)t}\big\| \ket{\phi'_{\bxp,\by}}\big\|^2 \\
%&= 2^{(n+i)t} \Big(2^{(n-i)t}-\Dist(n-i;t)\Big)\\
&= N^{2t} \cdot \negl(\lambda)\;.
\end{align*}
%where for the second line we used that the $\ket{\phi'_{\bxp,\by}}$ are normalized (see Claim~\ref{clm5}) and for the last line we used Lemma~\ref{lem:dist} and the assumption $n-i=\omega(\log\lambda)$. 
It remains to bound 
\begin{align}
\big\|  \ket{\tilde{\psi}_{Dist}} -  {\ket{\tilde{\psi}_{all}'}}\big\|^2 &= \sum_{\substack{\bxp,\by\\ \by_{>}\in \Dist(n-i;t)}} \Big\|\sum_{\substack{\bxpp \\\bxpp\in \Dist(n-i;t)\\ (\bxp,\bxpp,\by)\notin \Dist}} (-1)^{\sum_{j=1}^t y_j\cdot (x''_j 0^i)} \ket{e_{x'_1x''_1}\oplus e_{y_1}\oplus \cdots \oplus e_{x'_tx''_t}\oplus e_{y_t}}\Big\|^2\notag\\
&=  \sum_{\substack{\bxp,\by\\ \by_{>}\in \Dist(n-i;t)}} \Big\|\sum_{\substack{\bxpp \\\bxpp\in \Dist(n-i;t)\\ (\bxp,\bxpp,\by)\notin \Dist}} (-1)^{\sum_{j=1}^t y_j\cdot (x''_j 0^i)} \ket{e_{x'_1x''_1}\oplus  \cdots \oplus e_{x'_tx''_t}}\Big\|^2
\label{eq:clm2-2}\;.
\end{align}
Here, the second line is justified because for fixed $\by$, the operation $\ket{e}\mapsto \ket{e\oplus y_1\oplus\cdots \oplus y_t}$ is unitary. We now apply Lemma~\ref{lem:helper}, with the choice $\ket{u_{T,\by}} = \ket{e_{T_1}\oplus\cdots\oplus e_{T_{|T|}}}$ (which is independent of $\by$) and similarly $\ket{v_S}=\ket{e_{S_1}\oplus\cdots\oplus e_{S_{|S|}}}$. The norm condition is clearly satisfied for these vectors; as well as the orthogonality condition. Moreover, if $T=\{x_1,\ldots,x_t\}\backslash \{y_1,\ldots,y_t\}$ and $S= \{x_1,\ldots,x_t\}\cap \{y_1,\ldots,y_t\}$ then for fixed $\by$ and using that $\bx\in\Dist(n;t)$ there is a unitary map from $\ket{e_{x_1}\oplus  \cdots \oplus e_{x_t}} \mapsto \ket{u_T}\ket{v_S}$. Thus
\begin{align}
\eqref{eq:clm2-2} &= \sum_{\substack{\bxp,\by\\ \by_{>}\in \Dist(n-i;t)}} \Big\|\sum_{\substack{\bxpp \\\bxpp\in \Dist(n-i;t)\\ (\bxp,\bxpp,\by)\notin \Dist}} (-1)^{\sum_{j=1}^t y_j\cdot (x''_j 0^i)} \ket{u_{\{x_1,\ldots,x_t\}\backslash\{y_1,\ldots,y_t\},\by}}\ket{v_{\{x_1,\ldots,x_t\}\cap \{y_1,\ldots,y_t\}}}\Big\|^2\notag \\
&\leq \sum_{\substack{\bxp,\by\\ \by\in \Dist(n;t)}} \Big\|\sum_{\substack{\bxpp \\ (\bxp,\bxpp,\by)\notin \Dist}} (-1)^{\sum_{j=1}^t y_j\cdot (x''_j 0^i)} \ket{u_{\{x_1,\ldots,x_t\}\backslash\{y_1,\ldots,y_t\},\by}}\ket{v_{\{x_1,\ldots,x_t\}\cap \{y_1,\ldots,y_t\}}}\Big\|^2\;,\notag
\end{align}
where for the second line we relaxed the condition $\by_{>}\in \Dist(n-i;t)$ to $\by\in \Dist(n;t)$, and dropped the condition $\bxpp\in\Dist(n-i;t)$. The latter is possible because vectors $\ket{u_{\{x_1,\ldots,x_t\}\backslash\{y_1,\ldots,y_t\},\by}}\ket{v_{\{x_1,\ldots,x_t\}\cap \{y_1,\ldots,y_t\}}}$ are orthogonal when $\bxpp\in\Dist(n-i;t)$ or when $\bxpp\notin\Dist(n-i;t)$ (because the set $\{x''_1,\ldots,x''_t\}$ can be read from them and its cardinality evaluated), so dropping the condition can only increase the total squared norm. Applying Lemma~\ref{lem:helper} concludes the claim. 
\end{proof}

Using Claims \ref{clm5}, \ref{clm2}, Fact \ref{fact2} and that partial trace is a 
trace-preserving-completely-positive map, we get:
\begin{equation}
\label{exp7} 
    \E_{f\in \{0,1\}^N}[\ketbra{\psi_{f}}{\psi_{f}}^{\otimes t}] \underset{\ref{clm5}}{=} \Tr_E\left[ \ket{\psi_{\textit{all}}}\right] \overset{s}{\underset{\ref{clm2}, \ref{fact2}}{\approx}} 
    \Tr_E\left[ \ket{\psi_{Dist}}\right]\;.
\end{equation}

\begin{claim}
\label{clm6} 
There is a set $\Good \subseteq (\{0,1\}^i)^t \times (\{0,1\}^n)^t$ such that
    \begin{equation*}
        \Tr_E\big[ \ket{\psi_{\Dist}}\big] \overset{s}{\approx}  Tr_E\Big[ 
        \frac{1}{\sqrt{|\Good|}}\sum_{(\mathbf{x'},\mathbf{y}) \in \Good} \ket{x'_1}\ket{y_1} \otimes ... \otimes \ket{x'_t}\ket{y_t} \bigotimes \ket{\{ x'_1 y_1, ... x'_ty_t\}} \Big]\;.
    \end{equation*}

\end{claim}

\begin{proof}
To prove this claim we are going to use three subclaims (\ref{fermi1}, \ref{fermi1b}, \ref{close to Good}) that will help us with 
the approximation.
First, let's analyze $\ket{\psi_{\Dist}}$. Recall that
\begin{equation}
\label{psi dist}
    \ket{\psi_{\Dist}} \,=\, \frac{1}{\sqrt{\abs{\Dist}}} \sum_{\mathbf{x', x'', y} \in \Dist} \ket{\phi_{\mathbf{x'},\mathbf{y}}}\otimes (-1)^{\sum_{j=1}^t y_j\cdot (x''_j 0^i)} \ket{e_{x'_1x''_1}\oplus e_{y_1}...\oplus e_{x'_tx''_t}\oplus e_{y_t}}\;.
\end{equation}
Since we sum only over elements such that $(x'_1x''_1, y_1, ...,x'_tx''_t, y_t) \in \Dist$, there exists an isometry between $\ket{e_{x'_1x''_1}\oplus e_{y_1}...\oplus e_{x'_tx''_t}\oplus e_{y_t}}$ and 
$\ket{\{x'_1x''_1, y_1, ...,x'_tx''_t, y_t\}}$. Here, letting $R_\pi\ket{x_1,...,x_{2t}} \deff \ket{x_{\pi(1)},...,x_{\pi(2t)}}$ 
for any $\pi\in\mathfrak{S}_{2t}$, we set 
\begin{equation}
\label{set ket}
    \ket{\{x'_1x''_1, y_1, ...,x'_tx''_t, y_t\}} \deff \frac{1}{\sqrt{(2t)!}} \sum_{\pi \in \mathfrak{S}_{2t}} R_\pi \ket{x'_1x''_1, y_1, ...,x'_tx''_t, y_t}\;.
\end{equation}
Note that if $(\bxp,\bxpp,\by)\notin \Dist$ then the state $\ket{\{x'_1x''_1, y_1, ...,x'_tx''_t, y_t\}}$ is not necessarily normalized, see Lemma~\ref{lem:rpi2}.

Here however the elements \emph{are} distinct, so 
using again Claim \ref{clm1} we can apply the isometry to the purifying register and still get a purification for the original state, i.e.
\begin{equation}
\label{exp4}
    \begin{aligned}
    \Tr_E\big[ \ket{\psi_{Dist}}\big] &\underset{\eqref{psi dist}}{=} \Tr_E\Big[ \frac{1}{\sqrt{\abs{\Dist}}} \sum_{\mathbf{x', x'', y} \in \Dist} \ket{\phi_{\mathbf{x'},\mathbf{y}}}\otimes (-1)^{\sum_{j=1}^t y_j\cdot (x''_j 0^i)} \ket{e_{x'_1x''_1}\oplus e_{y_1}...\oplus e_{x'_tx''_t}\oplus e_{y_t}}\Big]\\
        &\underset{\ref{clm1}}{=} Tr_E\Big[ \frac{1}{\sqrt{\abs{\Dist}}} \sum_{\mathbf{x', x'', y} \in\Dist} \ket{\phi_{\mathbf{x'},\mathbf{y}}}\otimes (-1)^{\sum_{j=1}^t y_j\cdot (x''_j 0^i)} \ket{\{x'_1x''_1, y_1, ...,x'_tx''_t, y_t\}}\Big]\;.
    \end{aligned}
\end{equation}
Denote 
\[\forall 1 \le j \le t\: \qquad y_j \,=\, \underbrace{y_j'}_{\in \{0,1\}^{n-i}}  \underbrace{y_j''}_{\in \{0,1\}^{i}}\;,\] 
and $\mathbf{y'} \deff$ $(y_1',...,y_t')$. This is convenient because $y_j\cdot (x''_j 0^i) = y_j' \cdot x''_j$. For the next claim, recall that for every $j$, $y_{j,>}$ denotes the last $(n-i)$ bits of $y_j$.

\begin{subclaim}\label{fermi1}
\begin{align*}
    \Tr_E\Big[ \frac{1}{\sqrt{|\Dist|}}&\sum_{(\mathbf{x', x'', y}) \in \Dist} \ket{\phi_{\mathbf{x'},\mathbf{y}}} \otimes (-1)^{\sum_{j=1}^t y_j'\cdot x''_j} \ket{\{x'_1x''_1, y_1, \ldots, x'_tx''_t, y_t\}}\Big] \\
&    \overset{s}{\approx}
    \Tr_E\Big[ \frac{1}{2^{tn}}\sum_{\substack{\bxp, \bxpp, \by\\\by\in\Dist(n;t)}} \ket{\phi_{\mathbf{x'},\mathbf{y}}} \otimes (-1)^{\sum_{j=1}^t y'_j\cdot x''_j } \ket{\{x'_1x''_1, y_1, \ldots, x'_tx''_t, y_t\}}\Big]\;.
\end{align*}
\end{subclaim}

\begin{proof}
The proof follows from Lemma~\ref{lem:helper}. In slightly more detail, first note that due to~\eqref{eq:dist-size} the normalization by $1/\sqrt{|\Dist|}$ for the first state can be replaced by $1/2^{tn}$ at the cost of a negligible change in trace distance. Next, evaluating the squared norm of the difference, 
\begin{align*}
  \frac{1}{N^{2t}}\Big\|  &\sum_{\substack{(\mathbf{x', x'', y}) \notin \Dist\\\by\in\Dist(n;t)}}  \ket{\phi_{\mathbf{x'},\mathbf{y}}} \otimes (-1)^{\sum_{j=1}^t y'_j\cdot x''_j} \ket{\{x'_1x''_1, y_1, ...,x'_tx''_t, y_t\}} \Big\|^2 \notag \\
 % &= \sum_{\bxp,\by} \Big\| \sum_{\bxpp: (\bxp,\bxpp,\by)\notin \Dist} \otimes (-1)^{\sum_{j=1}^t y_j\cdot (x''_j 0^i)} \ket{\{x'_1x''_1,\ldots,x'_tx''_t\}} \Big\|^2 \\
  &= \frac{1}{N^{2t}} \sum_{\substack{\bxp,\by\\\by\in\Dist(n;t)}} \Big\| \sum_{\substack{\bxpp\\ (\mathbf{x', x'', y}) \notin \Dist}}  (-1)^{\sum_{j=1}^t y'_j\cdot x''_j} \ket{\{x'_1x''_1, y_1, ...,x'_tx''_t, y_t\}}\Big\|^2 \;. 
\end{align*}
This expression is bounded by Lemma~\ref{lem:helper}, where we let $\ket{u_{T,\by}} = \ket{T\cup\{y_1,\ldots,y_t\}}$, which satisfies the desired conditions by Lemma~\ref{lem:rpi2}.
\end{proof}

Overall, starting from~\eqref{exp4} and applying Claim~\ref{fermi1}.
\begin{align} \label{dist close to evrything}
    \Tr_E\left[ \ket{\psi_{Dist}}\right] &\underset{\eqref{exp4}}{=}\Tr_E\Big[ \frac{1}{\sqrt{\abs{\Dist}}} \sum_{\mathbf{x', x'', y} \in \Dist} \ket{\phi_{\mathbf{x'},\mathbf{y}}}\otimes (-1)^{\sum_{j=1}^t y_j'\cdot x''_j } \ket{\{x'_1x''_1, y_1, ...,x'_tx''_t, y_t\}}\Big] \nonumber\\
        &\overset{s}{\underset{\ref{fermi1}}{\approx}} \Tr_E\Big[ \frac{1}{2^{tn}}\sum_{\substack{ \mathbf{x', x'', y} \\\by\in\Dist(n;t)}} \ket{\phi_{\mathbf{x'},\mathbf{y}}} \otimes (-1)^{\sum_{j=1}^t y'_j\cdot x''_j } \ket{\{x'_1x''_1, y_1, ...,x'_tx''_t, y_t\}}\Big]\;.
\end{align}

Let's focus on the state in the partial trace. Recall that according to \eqref{phi} $\ket{\phi_{\mathbf{x'},\mathbf{y}}}$ is  
$H(\ket{x'_1}\ket{y_1} \otimes \cdots \otimes \ket{x'_t}\ket{y_t})$. Note that the sum is also over $\mathbf{x''}$ but $\mathbf{x''}$ doesn't  appear in $\ket{\phi_{\mathbf{x'},\mathbf{y}}}$, 
so we can break the sum in the partial trace in Eq.~\eqref{dist close to evrything} into two:
\begin{equation*}
        \frac{1}{2^{tn}}\sum_{\substack{\bxp, \bxpp, \by\\\by\in\Dist(n;t)}} \ket{\phi_{\mathbf{x'},\mathbf{y}}} \otimes (-1)^{\sum_{j=1}^t y'_j\cdot x''_j} \ket{\{x'_1x''_1, y_1, \ldots,x'_tx''_t, y_t\}}
\end{equation*}
\begin{equation*}
        =
\end{equation*}
\begin{equation}
\label{exp5}
        \frac{1}{\sqrt{2^{n + i}}^t}\sum_{\substack{\mathbf{x'}, \mathbf{y} \\ \by\in\Dist(n;t)}} \ket{\phi_{\mathbf{x'},\mathbf{y}}}\otimes
        \underbrace{\frac{1}{\sqrt{2^{(n-i)}}^t}\sum_{\mathbf{x''}} (-1)^{\sum_{j=1}^t y'_j\cdot x''_j} \ket{\{x'_1x''_1, y_1, \ldots, x'_tx''_t, y_t\}}}_{\ket{\phi''_{\mathbf{x',y}}}}\;.
\end{equation}
Clarifying $\ket{\phi''_{\mathbf{x',y}}}$:
\begin{align} 
    \ket{\phi''_{\mathbf{x',y}}} &= \frac{1}{\sqrt{2^{(n-i)}}^t} \sum_{\mathbf{x''}} (-1)^{\sum_{j=1}^t y'_j\cdot x''_j} \ket{\{x'_1x''_1, y_1, \ldots,x'_tx''_t, y_t\}} \nonumber \\
    &\underset{\eqref{set ket}}{=} \frac{1}{\sqrt{2^{(n-i)}}^t} \frac{1}{\sqrt{(2t)!}} \sum_{\mathbf{x''}} (-1)^{\sum_{j=1}^t y'_j\cdot x''_j} \sum_{\pi\in S_{2t}} R_\pi \ket{x'_1x''_1, y_1,\ldots,x'_tx''_t, y_t} \nonumber \\
    &= \frac{1}{\sqrt{(2t)!}} \sum_{\pi\in S_{2t}} R_\pi \frac{1}{\sqrt{2^{(n-i)}}^t}\sum_{\mathbf{x''}} (-1)^{\sum_{j=1}^t y'_j\cdot x''_j} \ket{x'_1x''_1, y_1,\ldots,x'_tx''_t, y_t} \nonumber \\
    &= \frac{1}{\sqrt{(2t)!}} \sum_{\pi\in S_{2t}} R_\pi \frac{1}{\sqrt{2^{(n-i)}}}\sum_{x''_1} (-1)^{y_1'\cdot x''_1 } \ket{x'_1x''_1, y_1} \otimes \cdots
    \otimes \frac{1}{\sqrt{2^{(n-i)}}}\sum_{x''_t} (-1)^{y'_t\cdot x''_t} \ket{x'_tx''_t, y_t} \nonumber \\
    &= \frac{1}{\sqrt{(2t)!}} \sum_{\pi\in \mathfrak{S}_{2t}} R_\pi \Big( \underset{j=1}{\overset{t}{\bigotimes}} \ket{x'_j} \frac{1}{\sqrt{2^{(n-i)}}}\sum_{x''_j} (-1)^{y'_j\cdot x''_j} \ket{x''_j} \ket{y_j} \Big) \nonumber\\
 &= \frac{1}{\sqrt{(2t)!}} \sum_{\pi\in \mathfrak{S}_{2t}} R_\pi \Big( \underset{j=1}{\overset{t}{\bigotimes}} \ket{x'_j} H\ket{y_j'} \ket{y_j} \Big)\;,\label{phi'}
\end{align}
where the last equality uses
\begin{equation*}\label{H} 
    \frac{1}{\sqrt{2^{n-i}}}\sum_{x''_j} (-1)^{y'_j\cdot x''_j} \ket{x''_j} = H\ket{y'_j}\;.
\end{equation*}
With this simplified formulation, it is not hard to remove the restriction  $\by\in\Dist(n;t)$, as is shown in the following claim. 

\begin{subclaim}\label{fermi1b}
\begin{align*} 
      \frac{1}{\sqrt{2^{n+i}}^t}&\sum_{\substack{\bxp, \by\\\by\in\Dist(n;t)}} \ket{\phi_{\mathbf{x'},\mathbf{y}}} \otimes\ket{\phi''_{\mathbf{x',y}}} \,\overset{s}{\approx}\,
    \frac{1}{\sqrt{2^{n+i}}^t}\sum_{\substack{\bxp, \by}} \ket{\phi_{\mathbf{x'},\mathbf{y}}} \otimes \ket{\phi''_{\mathbf{x',y}}}\;.
\end{align*}
\end{subclaim}

\begin{proof}
We first evaluate the norm of $\ket{\phi''_{\bx',\by}}$, using an argument similar to the proof of Lemma~\ref{lem:rpi2}. We have
\begin{align}
\big\| \ket{\phi''_{\bx',\by}} \big\|^2
&= \frac{1}{(2t)!} \sum_{\pi,\pi'\in \mathfrak{S}_{2t}}\Big( \bigotimes_{j=1}^t  \bra{x'_j} \bra{y_j'}H \bra{y_j} \Big) R_{\pi'}R_\pi \Big( \underset{j=1}{\overset{t}{\bigotimes}} \ket{x'_j} H\ket{y_j'} \ket{y_j} \Big)\notag\\
&= \sum_{\pi\in \mathfrak{S}_{2t}}\Big( \bigotimes_{j=1}^t  \bra{x'_j} \bra{y_j'}H \bra{y_j} \Big) R_\pi \Big( \underset{j=1}{\overset{t}{\bigotimes}} \ket{x'_j} H\ket{y_j'} \ket{y_j} \Big)\;.\label{eq:fermi1b-1}
\end{align}
Now we observe the following. First, for any $j,j'\in\{1,\ldots,t\}$, $|\bra{x_j}\bra{y'_j}H \cdot \ket{y_{j'}}|\leq 2^{-(n-i)/2}$. Second, for $j,j'\in\{1,\ldots,t\}$ we have that $|\bra{x_j}\bra{y'_j}H \cdot \ket{x_{j'}}H\ket{y'_{j'}}|$ is zero unless $y'_j=y'_{j'}$, in which case it is at most $1$; similarly, $|\bra{y_j} \cdot \ket{y_{j'}}|$ is zero unless $y_j=y_{j'}$, in which case it is at most $1$. Summarizing rather crudely, if we let $\ket{z_j} = \ket{x_j}H\ket{y'_j}$ if $j\in\{1,\ldots,t\}$ and $\ket{z_j}=\ket{y_j}$ if $j\in\{t+1,\ldots,2t\}$ then 
\begin{equation}\label{eq:fermi1b-2}
\big|\bra{z_j} z_{j'}\rangle\big|\,\leq\, 2^{-\frac{n-i}{2}}\delta_{y'_{j\mod t}\neq y'_{j'\mod t}} + \delta_{y'_{j\mod t}=y'_{j'\mod t}}\;.
\end{equation}
For $j\in\{1,\ldots,2t\}$ let $\ell_j$ denote the number of $j'\in\{1,\ldots,2t\}$ such that ${y'_{j\mod t}=y'_{j'\mod t}}$. Then, bounding the sum over all permutations in~\eqref{eq:fermi1b-1} by the sum over all functions, 
\begin{align*}
\eqref{eq:fermi1b-1} &\leq \sum_{f:[2t]\to[2t]} \prod_{j=1}^{2t} \big|\bra{z_{f(j)}} z_{j}\rangle\big|\\ 
&\leq \prod_{j=1}^{2t} \Big( \sum_{j'=1}^{2t} \big|\bra{z_{j'} }z_{j}\rangle\big|\Big)\\
&\leq \prod_{j=1}^{2t} \Big( \ell_{j} + \frac{t}{2^{\frac{n-i}{2}}}\Big)\\
&\leq \Big(\prod_{j=1}^t \ell_j\Big)^2\Big(1 + \frac{O(t^2)}{2^{\frac{n-i}{2}}}\Big)\;.
\end{align*}
Here, the third line is by~\eqref{eq:fermi1b-2} and the definition of $\ell_j$. Finally, note that $\prod_{j=1}^t \ell_j \leq 2^{t-k}$ where $k=|\{y'_1,\ldots,y'_t\}|\leq |\{y_1,\ldots,y_t\}|<t$, which gives us the bound
\begin{equation}\label{eq:fermi1b-3}
 \big\| \ket{\phi''_{\bx',\by}} \big\|^2\,\leq\, 2^{2(t-k)}\Big(1 + \frac{O(t^2)}{2^{\frac{n-i}{2}}}\Big)\;.
\end{equation}
We now conclude the proof:
\begin{align*}
 \Big\| \frac{1}{\sqrt{2^{n+i}}^t}\sum_{\substack{\bxp, \by\\\by\notin\Dist(n;t)}} \ket{\phi_{\mathbf{x'},\mathbf{y}}} \otimes\ket{\phi''_{\mathbf{x',y}}}\Big\|^2
&=  \frac{1}{2^{(n+i)t}}\sum_{\substack{\bxp, \by\\\by\notin\Dist(n;t)}} \big\|\ket{\phi''_{\mathbf{x',y}}}\big\|^2\\
&\leq \frac{1}{2^{(n+i)t}} 2^{2it} \sum_{k=1}^{t-1} \sum_{\byp:\ |\{y'_1,\ldots,y'_t\}|=k} 2^{2(t-k)}\Big(1 + \frac{O(t^2)}{2^{\frac{n-i}{2}}}\Big)\\
&= \frac{1}{2^{(n-i)t}} \sum_{k=1}^t {2^{n-i}\choose k} 2^{2(t-k)}\Big(1 + \frac{O(t^2)}{2^{\frac{n-i}{2}}}\Big)\\
&\leq  \sum_{k=1}^t  \Big(\frac{4}{2^{n-i}}\Big)^{t-k}\Big(1 + \frac{O(t^2)}{2^{\frac{n-i}{2}}}\Big)\\
&= \negl(\lambda)\;,
\end{align*}
where for the second line we used~\eqref{eq:fermi1b-3} and for the last line we used $n-i=\omega(\log\lambda)$ and $t=\poly(\lambda)$. 
\end{proof}

Given the structure of the state in \eqref{exp5}, and considering Claim~\ref{fermi1b}, we can use 
Claim \ref{clm3} to move the $H = H^T$ gate that acts on $\ket{x_j'}, \ket{y_j}$ with transpose to their 
corresponding parts in the purifying register:
\begin{align} \label{exp12} 
    \eqref{exp5} &= \frac{1}{\sqrt{2^{n+i}}^t}\sum_{\substack{\bxp, \by\\\by\in\Dist(n;t)}} \ket{\phi_{\mathbf{x'},\mathbf{y}}} \otimes\ket{\phi''_{\mathbf{x',y}}} \nonumber\\
    &\underset{\ref{fermi1b}}{\overset{s}{\approx}}\frac{1}{\sqrt{2^{n + i}}^t}\sum_{\substack{\bxp,\by }} \ket{\phi_{\mathbf{x'},\mathbf{y}}} \otimes \ket{\phi''_{\mathbf{x',y}}} \nonumber\\
    &\underset{\eqref{phi}}{=} \frac{1}{\sqrt{2^{n + i}}^t}\sum_{\substack{\bxp,\by }} H(\ket{x'_1}\ket{y_1} \otimes \cdots \otimes \ket{x'_t}\ket{y_t}) \bigotimes \frac{1}{\sqrt{(2t)!}} \sum_{\pi\in S_{2t}} R_\pi \Big( \underset{j=1}{\overset{t}{\bigotimes}} \ket{x'_j} H\ket{y_j'} \ket{y_j} \Big) \nonumber\\
    &\underset{\ref{clm3}}{=} \frac{1}{\sqrt{2^{n + i}}^t}\sum_{\substack{\bxp,\by }} \ket{x'_1}\ket{y_1} \otimes \cdots \otimes \ket{x'_t}\ket{y_t} \bigotimes \frac{1}{\sqrt{(2t)!}} \sum_{\pi\in S_{2t}} R_\pi \Big( \underset{j=1}{\overset{t}{\bigotimes}} H\ket{x'_j}H\ket{y_j'} H\ket{y_j}  \Big) \nonumber\\
    &= \frac{1}{\sqrt{2^{n + i}}^t}\sum_{\substack{\bxp,\by }} \ket{x'_1}\ket{y_1} \otimes\cdots \otimes \ket{x'_t}\ket{y_t} \bigotimes \frac{1}{\sqrt{(2t)!}} \sum_{\pi\in S_{2t}} R_\pi H \Big( \underset{j=1}{\overset{t}{\bigotimes}} \ket{x'_j} \ket{y_j'}  \ket{y_j}  \Big) \nonumber \\
    &= \frac{1}{\sqrt{2^{n + i}}^t} \sum_{\substack{\bxp,\by }} \ket{x'_1}\ket{y_1} \otimes \cdots \otimes \ket{x'_t}\ket{y_t} \bigotimes  H \frac{1}{\sqrt{(2t)!}} \sum_{\pi\in S_{2t}}  R_\pi \Big( \underset{j=1}{\overset{t}{\bigotimes}} \ket{x'_j} \ket{y_j'}  \ket{y_j}  \Big) \nonumber \\
    &\underset{\eqref{set ket}}{=} \frac{1}{\sqrt{2^{n + i}}^t} \sum_{\substack{\bxp,\by }} \ket{x'_1}\ket{y_1} \otimes \cdots \otimes \ket{x'_t}\ket{y_t} \bigotimes  H \ket{\{x'_1y'_1, y_1, \ldots, x'_ty'_t, y_t\}}\;.
\end{align}

Now we are going to define the set $\Good$.
For future purposes we need to define the set $\Good$ such that there is a bijection 
between $\{x'_1y'_1, y_1, \ldots, x'_ty'_t, y_t\}$ and $\{x_1'y_1, \ldots, x_t'y_t\}$ for every $(\mathbf{x', y}) \in \Good$.
Note that the first set has $2t$ elements and the second has $t$ elements of longer size.
We want to see when there exists a bijection between one set to the other.
It is easy to verify that there exists a one to one map from $\{x_1'y_1, \ldots, x_t'y_t\}$ to $\{x'_1y'_1, y_1, \ldots, x'_ty'_t, y_t\}$.
We want to see when there exists also a one to one map from $\{x'_1y'_1, y_1, \ldots, x'_ty'_t, y_t\}$ to 
$\{x_1'y_1, \ldots, x_t'y_t\}$. This could be done by matching the correct couples of elements i.e $x_j'y_j'$ to 
the correct couple element $y_j$ and concatenate $x_j'y_j$. 
For this consider the following definition:
\begin{equation}\label{Good def}
    (\mathbf{x', y}) \in \Good \iff 
    %\exists \mathbf{x''}:\;\mathbf{x', x'', y} \in \Dist \wedge
    \mathbf{y'} \in \Dist(n-i;t) \wedge     \mathbf{y} \in G\;,
\end{equation}
where 
%$\Dist$ is as defined in \eqref{Dist def}, 
$\Dist(n-i;t)$ is as defined in Definition~\ref{def:dist-unique}  and $G = \prod_j G_j$ where $y_j\in G_j$ where $G_j$ is the set of all $y_j$ such that if we write
$y_j = y_j^0 \cdots y_j^i y_j^{i+1}\cdots y_j^{n-1}$ then $ y_j^{2i}\cdots y_j^{n-1} \neq y_j^0\cdots y_j^{n-2i-1}$, 
i.e\ the $(n-2i)$-suffix of $y_j$ does not equal its $(n-2i)$-prefix. 

\phantomsection \label{isometry}
With this definition of the set $\Good$ it is not hard to verify that for every $\mathbf{x',y} \in \Good$ there exists a bijection between $\{x'_1y'_1, y_1, \ldots, x'_ty'_t, y_t\}$ and $ \{x_1'y_1, \ldots, x_t'y_t\}$. In particular, the condition $\by\in G$ is used to guarantee that elements of the form $y_j$ and elements of the form $x'_jy'_j$ cannot be confused with each other. Note also that $(\bxp,\by)\in \Good$ implies that $\{x'_1y'_1,y_1,\ldots,x'_ty'_t,y_t\}\in \Dist$. 

\begin{subclaim}\label{close to Good}
    \begin{align*}
        \Tr_E\Big[ & \frac{1}{\sqrt{2^{n + i}}^t} \sum_{\substack{\bxp,\by }} \ket{x'_1}\ket{y_1} \otimes \cdots \otimes \ket{x'_t}\ket{y_t} \bigotimes \ket{\{ x'_1y'_1, y_1, \ldots, x'_ty'_t, y_t\}} \Big]\\
&        \overset{s}{\approx}
        \Tr_E\Big[ \frac{1}{\sqrt{|\Good|}} \sum_{(\mathbf{x',y}) \in \Good} \ket{x'_1}\ket{y_1} \otimes \cdots \otimes \ket{x'_t}\ket{y_t} \bigotimes \ket{\{ x'_1y'_1, y_1, \ldots, x'_ty'_t, y_t\}} \Big]\;.
    \end{align*}      
\end{subclaim}

\begin{proof}
We first evaluate $|\Good|$. %We already saw~\eqref{eq:dist-size} that $|\Dist| \geq N^{2t}(1-O(t^2)/N)$, and 
By Lemma~\ref{lem:dist}, $|\Dist(n-i;t)|\geq 2^{(n-i)t}(1-O(t^2)/2^{n-i})$. Regarding $G$, note that for any $1\leq j\leq t$ and $y\in\{0,1\}^n$,
\[\Pr[y \notin G_j] = \Pr[y^{2i}\cdots y^{n-1} = y^0\cdots y^{n-2i-1}] = \frac{1}{2^{n-i}}\;,\] 
which is negligible due to $i = n - \omega(\log(\lambda))$. Thus
\[\Pr[\mathbf{y} \in G] = \Pr[y_1,...y_t \in G'] = \Pi_{j=1}^t \Pr[y_j\in G_j] = (1- \negl(\lambda))^t = 1 - \negl(\lambda)\;.\] 
Overall, $|\Good|\geq N^{2t}(1-\negl(\lambda)$. 

Using the fact that the vectors $\ket{x'_1}\ket{y_1} \otimes \cdots \otimes \ket{x'_t}\ket{y_t} \bigotimes \ket{\{ x'_1y'_1, y_1, \ldots, x'_ty'_t, y_t\}}$ are pairwise orthogonal it suffices to bound
\begin{align*}
      \Big\| \frac{1}{2^{\frac{(n+i)t}{2}}} & \sum_{\substack{ (\mathbf{x',y}) \notin \Good}} \ket{x'_1}\ket{y_1} \otimes \cdots \otimes \ket{x'_t}\ket{y_t} \bigotimes \ket{\{ x'_1y'_1, y_1, \ldots, x'_ty'_t, y_t\}} \Big]\Big\|^2\\
      &= \frac{1}{2^{(n+i)t}} \sum_{\substack{ (\mathbf{x',y}) \notin \Good}} \big\|\ket{\{ x'_1y'_1, y_1, \ldots, x'_ty'_t, y_t\}} \big\|^2\\
      &\leq \frac{1}{2^{(n-i)t}}\sum_{k=1}^t {t\choose k}\sum_{k'=0}^k 2^{(n-i)k'}2^{i(k-k')}t^{t-k}(2t-(k+k')+1)!\\
      &\leq t^2 \frac{t^{2(t-k)} t^{2t-(k+k')+1}}{2^{(n-i)(t-k')}}\\
      &\leq \negl(\lambda)\;.
\end{align*}
Here, the third line is justified as follows. The integer $k$ denotes $k=|\{y'_1,\ldots,y'_t\}|$, and $k'$ denotes, among $y'_1,\ldots,y'_t$, the number of elements that are not in $G$. By definition of $\Good$, necessarily $k<t$ or $k'>0$. Furthermore, given $k,k'$ we have that $|\{ x'_1y'_1, y_1, \ldots, x'_ty'_t, y_t\}| = (k-k')+2k'=k+k'$, giving the bound $(2t-(k+k')+1)!$ on $\ket{\{ x'_1y'_1, y_1, \ldots, x'_ty'_t, y_t\}}$ by Lemma~\ref{lem:rpi2}. Finally given $k,k'$ there are at most ${t\choose k}2^{(n-i)k'}$ choices for the distinct $y'_j$ such that $y'_j\notin G_j$, and at most $2^{ik'}$ for the $y'_j\in G_j$; and given these, at most $t^{t-k}$ for the remaining $y'_j$. The last line uses $k'\leq k$  and $t-k'>0$, because $k'=t$ implies $k=t$ as well, which is not possible, and the assumption $n-i=\omega(\log\lambda)$.
\end{proof}

Finally,
\begin{align*}
    \Tr_E\left[ \ket{\psi_{Dist}}\right] &\underset{\eqref{dist close to evrything}}{\overset{s}{\approx}} \Tr_E\Big[ \frac{1}{2^{tn}}\sum_{\substack{ \mathbf{x', x'', y} \\\byp\in\Dist(n-i;t)}} \ket{\phi_{\mathbf{x'},\mathbf{y}}} \otimes (-1)^{\sum_{j=1}^t y'_j\cdot x''_j } \ket{\{x'_1x''_1, y_1, ...,x'_tx''_t, y_t\}}\Big]\\
    &\underset{\eqref{exp5}}{=} \Tr_E\Big[ \frac{1}{\sqrt{2^{n + i}}^t}\sum_{\substack{\mathbf{x'}, \mathbf{y} \\ \byp\in\Dist(n-i;t)}} \ket{\phi_{\mathbf{x'},\mathbf{y}}} \bigotimes \ket{\phi''_{\mathbf{x',y}}} \Big] \\
    &\underset{\eqref{exp12}}{=} \Tr_E\Big[ \frac{1}{\sqrt{2^{n + i}}^t} \sum_{\mathbf{x'}, \mathbf{y}} \ket{x'_1}\ket{y_1} \otimes ... \otimes \ket{x'_t}\ket{y_t} \bigotimes  H \ket{\{x'_1y'_1, y_1, ... x'_ty'_t, y_t\}} \Big]\\
    &= \Tr_E\Big[ \frac{1}{\sqrt{2^{n + i}}^t} \sum_{\mathbf{x'}, \mathbf{y}} \ket{x'_1}\ket{y_1} \otimes ... \otimes \ket{x'_t}\ket{y_t} \bigotimes  \ket{\{x'_1y'_1, y_1, ... x'_ty'_t, y_t\}} \Big]\\
    &\overset{s}{\underset{\ref{close to Good}}{\approx}} \Tr_E\Big[ \frac{1}{\sqrt{\Good}}\sum_{(\mathbf{x',y}) \in \Good} \ket{x'_1}\ket{y_1} \otimes ... \otimes \ket{x'_t}\ket{y_t} \bigotimes \ket{\{ x'_1y'_1, y_1, ... x'_ty'_t, y_t\}} \Big]\\
    &= \Tr_E\Big[ \frac{1}{\sqrt{\Good}}\sum_{(\mathbf{x',y}) \in \Good} \ket{x'_1}\ket{y_1} \otimes ... \otimes \ket{x'_t}\ket{y_t} \bigotimes \ket{\{x_1'y_1, ... x_t'y_t\}} \Big]\;,
\end{align*}
where the fourth and the last lines is because we can use claim \ref{clm1} and apply the unitary $H$ and the isometry we know exists (from \ref{isometry}) to the purifying register.
This ends the proof of Claim \ref{clm6}.
\end{proof}

Recall that our goal is to prove $\E_{f\in \{0,1\}^N}[\ketbra{\psi_{f}}{\psi_{f}}^{\otimes t}] \overset{c}{\approx} \E_{\mu\leftarrow Haar}[\ketbra{\mu}{\mu}^{\otimes t}]$.
To this end, consider the following claim.

\begin{claim}
\label{clm4}
    For every unitary $U \in U(\mathcal{H})$,
    \begin{equation*}
        \E_{f\in \{0,1\}^N}[(U\ketbra{\psi_{f}}{\psi_{f}}U^\dagger)^{\otimes t}] \overset{c}{\approx}
        \E_{f\in \{0,1\}^N}[\ketbra{\psi_{f}}{\psi_{f}}^{\otimes t}]\;.
    \end{equation*}
\end{claim}

\begin{proof}
    First let us prove the following subclaim.

\begin{subclaim}\label{fermi}
    \begin{align*}
        Tr_E \Big[ 
        \frac{1}{\sqrt{2^{n+i}}^t} &\sum_{\mathbf{x'},\mathbf{y}} \ket{x'_1y_1, ... ,x'_ty_t} \bigotimes  \ket{ \{x'_1 y_1, ... x'_ty_t\}} \Big] \\
&        \overset{s}{\approx}
         Tr_E \Big[ 
        \frac{1}{\sqrt{|\Good|}}\sum_{(\mathbf{x'},\mathbf{y}) \in \Good} \ket{x'_1y_1, ... ,x'_ty_t} \bigotimes \ket{ \{x'_1 y_1, ... x'_ty_t\}} \Big]\;.
    \end{align*}
\end{subclaim}

\begin{proof}
The proof of this claim is almost identical, and in fact somewhat easier, to the proof of Claim~\ref{close to Good}, and so we omit it. 
\end{proof}

\begin{align*}
    \E_{f\in \{0,1\}^N}[(U\ketbra{\psi_{f}}{\psi_{f}}U^\dagger)^{\otimes t}] \; &\overset{s}{\underset{\eqref{exp7}}{\approx}} \;
    \Tr_E\left[ U^{\otimes t} \ket{\psi_{Dist}}\right]\\
    &\overset{s}{\underset{\ref{clm6}}{\approx}} \Tr_E\Big[ 
        \frac{1}{\sqrt{|\Good|}} \sum_{(\mathbf{x'},\mathbf{y}) \in \Good} U^{\otimes t} (\ket{x'_1}\ket{y_1} \otimes ... \otimes \ket{x'_t}\ket{y_t}) \bigotimes \ket{\{ x'_1 y_1, ... x'_ty_t\}} \Big]\\
    &\underset{\eqref{set ket}}{=} \Tr_E \Big[ \frac{1}{\sqrt{|\Good|}}
    \sum_{(\mathbf{x'},\mathbf{y}) \in \Good} U^{\otimes t}(\ket{x'_1y_1, ... ,x'_ty_t}) \bigotimes \frac{1}{\sqrt{t!}}\sum_{\pi\in S_t} R_\pi\ket{ x'_1 y_1, ... x'_ty_t} \Big]\\ 
    &\overset{s}{\underset{\ref{fermi}}{\approx}} \Tr_E \Big[ 
    \frac{1}{\sqrt{2^{n+i}}^t}\sum_{\mathbf{x'},\mathbf{y}} U^{\otimes t}(\ket{x'_1y_1, ... ,x'_ty_t}) \bigotimes \frac{1}{\sqrt{t!}}\sum_{\pi\in S_t} R_\pi\ket{ x'_1 y_1, ... x'_ty_t} \Big]\\
    &\underset{\ref{clm3}}{=} \Tr_E \Big[ 
    \frac{1}{\sqrt{2^{n+i}}^t}\sum_{\mathbf{x'},\mathbf{y}} \ket{x'_1y_1, ... ,x'_ty_t} \bigotimes \frac{1}{\sqrt{t!}}\sum_{\pi\in S_t} R_\pi U^{T^{\otimes t}} (\ket{ x'_1 y_1, ... x'_ty_t}) \Big]\\
    &= \Tr_E \Big[ 
    \frac{1}{\sqrt{2^{n+i}}^t}\sum_{\mathbf{x'},\mathbf{y}} \ket{x'_1y_1, ... ,x'_ty_t} \bigotimes U^{T^{\otimes t}} \frac{1}{\sqrt{t!}}\sum_{\pi\in S_t} R_\pi \ket{ x'_1 y_1, ... x'_ty_t} \Big] \\
    &\underset{\eqref{clm1}}{=} \Tr_E \Big[ 
    \frac{1}{\sqrt{2^{n+i}}^t}\sum_{\mathbf{x'},\mathbf{y}} \ket{x'_1y_1, ... ,x'_ty_t} \bigotimes \frac{1}{\sqrt{t!}}\sum_{\pi\in S_t} R_\pi \ket{ x'_1 y_1, ... x'_ty_t} \Big] \\
    &\underset{\eqref{set ket}}{=} \Tr_E \Big[ 
    \frac{1}{\sqrt{2^{n+i}}^t}\sum_{\mathbf{x'},\mathbf{y}} \ket{x'_1y_1, ... ,x'_ty_t} \bigotimes  \ket{ \{x'_1 y_1, ... x'_ty_t\}} \Big]\\
    &\overset{s}{\underset{\ref{fermi}}{\approx}} \Tr_E \Big[ 
    \frac{1}{\sqrt{|\Good|}}\sum_{(\mathbf{x'},\mathbf{y}) \in \Good} \ket{x'_1y_1, ... ,x'_ty_t} \bigotimes \ket{ \{x'_1 y_1, ... x'_ty_t\}} \Big]\\
    &\underset{\ref{clm6}}{\overset{s}{\approx}} \Tr_E\left[ \ket{\psi_{Dist}}\right]\\
    &\underset{\eqref{exp7}}{\overset{s}{\approx}} \E_{f\in \{0,1\}^N}[\ketbra{\psi_{f}}{\psi_{f}}^{\otimes t}] \;. 
\end{align*}
This ends the proof of claim \ref{clm4}.
\end{proof}

If this is true for every unitary it is also true for an expectation over some unitaries.
Therefore
\begin{equation*}
    \E_{f\in \{0,1\}^N}[\ketbra{\psi_{f}}{\psi_{f}}^{\otimes t}] \underset{\ref{clm4}}{\overset{c}{\approx}}
    \E_{U \leftarrow Haar} [\E_{f\in \{0,1\}^N}[(U\ketbra{\psi_{f}}{\psi_{f}}U^\dagger)^{\otimes t}] ] = \E_{\ket{\psi}\leftarrow Haar} [\ketbra{\psi}{\psi}^{\otimes t}]\;,
\end{equation*}
as needed.

\section{Conclusions}

We showed in this work that one can expand the output length of some PRSs using the purification technique, and that one can additionally trade-off the depth (efficiency) of the circuit that generates the extended PRS with the length of the key used to generate it. Our constructions provide a method for combining \textit{any} PRS with general-phase PRSs to produce a longer PRS. To the best of our knowledge, nothing like this has been done before. 
In addition we proposed a condition we conjecture under which these expansions can be done in a black-box way. This leads the way to expand quantum pseudo-randomness in a black-box manner, with the same key length and the weakest possible assumption (i.e the only implicit assumption is, of course, the existence of a PRS in the first place (that satisfies condition \ref{condition}), which is potentially a weaker assumption than the existence of one-way functions and  is arguably the weakest meaningful assumption in use in cryptography today). Combining  with the observation by Brakerski and Shalit~\cite{BS24} (already mentioned in the introduction) this is best we can hope to achieve in terms of cryptographic assumptions.

%--- Section ---%
\section{Open Questions}

Our main goal was to understand what manipulations one can do on the length of the output state of a PRS, while retaining the PRS properties. In particular what can one do to expand the length of the output. Our work made some progress in understanding how to do this in a black-box way, but we arguably only scratched the surface of what could be done. The following are a few questions that can help us better understand the relationship between the different size PRSs and the manipulations to move between them.

\begin{enumerate}
    
    \item Is condition \ref{condition} truly sufficient for our constructions? Is this condition necessary? Are there any more PRS's that satisfy this condition? (Or is it a complicated way to describe the family of PRS's we already know this construction is expanding?)

    \item We do not know how to show security of the subset PRS construction presented in \cite{jeronimo2023subset} using the purification technique. Could the type of constructions we study be shown to expand subset PRS using a different proof approach, or are they inherently unsuitable for subset PRS?

    \item Our work demonstrates how different PRSs can be combined to a longer one, paving the way for exploring additional methods of combining PRSs. This raises the question: what other techniques exist for combining PRSs while preserving their essential properties?

\end{enumerate}

%-------------------------------------------
% Optional Contents
%-------------------------------------------

%--- Section ---%

%-------------------------------------------
% References
%-------------------------------------------

% Print bibliography
\bibliographystyle{alpha}
\bibliography{references}

%-------------------------------------------
% Appendix
%-------------------------------------------
% Activate the appendix in the doc
% from here on sections are numerated with capital letters 
\appendix

% Change equation numbering format to be sequential within sections in the appendix
\renewcommand\theequation{\Alph{section}\arabic{equation}} % Redefine equation numbering format
\counterwithin*{equation}{section} % Number equations within sections
\renewcommand\thefigure{\Alph{section}\arabic{figure}} % Redefine equation numbering format
\counterwithin*{figure}{section} % Number equations within sections
\renewcommand\thetable{\Alph{section}\arabic{table}} % Redefine equation numbering format
\counterwithin*{table}{section} % Number equations within sections

% %--- Section ---%
% \section{Some Notation}
% \lipsum[10]

% \subsection{Appendix subsection title here}
% As shown in Equation \ref{eq_a1}, the section number is inserted in the equation number.
% \lipsum[11]

% \begin{equation}
% Y_\infty = \left( \frac{m}{\textrm{GeV}} \right)^{-3}
%     \left[ 1 + \frac{3 \ln(m/\textrm{GeV})}{15}
%     + \frac{\ln(c_2/5)}{15} \right]
% \label{eq_a1}
% \end{equation}

% \subsection{Appendix subsection  title here}
% As shown in Table \ref{tab_a1}, the section number is inserted in the table number.
% \lipsum[12]

% \begin{table}[!ht]
% \caption{Sample table with three parts and five columns\label{tab_a1}}
% \begin{threeparttable}
% \begin{tabular*}{\columnwidth}{@{\extracolsep\fill}lllll@{\extracolsep\fill}}
% \toprule
% column 1 & column 2 & column 3 & column 4 & column 5\\
% \midrule
% row 1 & data 0 & data 1 & data 2 & data 3 \\
% row 2 & data 4 & data 5 & data 6 & data 7 \\
% row 3 & data 8 & data 9 & data 10 & data 11\\
% \bottomrule
% \end{tabular*}
% \end{threeparttable}
% \end{table}

% %--- Section ---%
% \section{Some More Notation}

% As shown in Figure \ref{fig:cat2}, the section number is inserted in the figure number.
% \lipsum[13]

% % figure b1
% \begin{figure}[!htbp]
% \centering
% \includegraphics[width=0.8\linewidth]{figures/cat_momo_2.jpg}
% \caption{\label{fig:cat2}This cat picture is located at the 'figures' folder.}
% \end{figure}

% \lipsum[14]

% \subsection{Appendix subsection title here}
% \lipsum[15]

% \end{appendices}

\end{document}